\ifdefined\SINGLECOL
\documentclass[12pt, draftclsnofoot, onecolumn]{IEEEtran}
\else
\documentclass[10pt]{IEEEtran}
\fi
\usepackage{amsmath}
\usepackage{amssymb}
\usepackage{amsthm}
\usepackage{mathtools}
\usepackage{bbm}
\usepackage{bm}
\usepackage{url}
\usepackage{caption}
\usepackage{dblfloatfix}
\usepackage{subcaption}
\usepackage{comment}
\usepackage{booktabs}
\usepackage[noadjust]{cite}
\usepackage{multirow}
\newtheorem{theorem}{Theorem}
\newtheorem{corollary}{Corollary}
\newtheorem{lemma}{Lemma}
\newtheorem{proposition}{Proposition}

\usepackage[capitalize]{cleveref}
\DeclareMathOperator*{\argmin}{arg\,min}
\DeclareMathOperator*{\argmax}{arg\,max}

\DeclareMathOperator{\E}{\mathbb{E}}
\DeclareMathOperator{\diag}{diag}
\usepackage{pgfplots}
\pgfplotsset{compat=1.15}
\pgfplotsset{colormap name=jet}
\usepgfplotslibrary{fillbetween}
\usetikzlibrary{patterns}
\usetikzlibrary{plotmarks}
\usetikzlibrary{trees}
\usepgfplotslibrary{groupplots}
\usetikzlibrary{automata,positioning,backgrounds}
\usetikzlibrary{external}
\usepackage{algorithm}
\usepackage{algpseudocode}

\begin{document}
\title{Timely Monitoring of Dynamic Sources with Observations from Multiple Wireless Sensors}

\author{Anders~E.~Kal{\o}r,~\IEEEmembership{Student Member,~IEEE,} and Petar~Popovski,~\IEEEmembership{Fellow,~IEEE}%
\thanks{The work has been supported by the Danish Council for Independent Research, Grant Nr. 8022-00284B SEMIOTIC and by the Digital Technologies for Industry 4.0 project at the TECH faculty, Aalborg University.}%
\thanks{A.~E.~Kal{\o}r and P. Popovski are with the Department of Electronic Systems, Aalborg University, 9220 Aalborg, Denmark (e-mail: aek@es.aau.dk; petarp@es.aau.dk).}}

\maketitle

\begin{abstract}
Age of Information (AoI) has recently received much attention due to its relevance for IoT sensing and monitoring. In this paper, we consider the problem of minimizing the AoI in a system in which a set of sources are observed by multiple sensors in a many-to-many relationship, and the probability that a sensor observes a source depends on the source's state. This model represents many practical scenarios, such as when multiple cameras or microphones are deployed to monitor objects moving in certain areas.
We formulate the scheduling problem as a Markov Decision Process, and show how the age-optimal scheduling policy can be obtained. We further consider partially observable variants of the problem, and devise approximate policies for large state spaces. The evaluations show that the approximate policies work well in the considered scenarios, while the fact that sensors can observe multiple sources is beneficial, especially when there is high uncertainty of the source states.
\end{abstract}

\section{Introduction}
Sensing and monitoring of the environment is a generic use case for the Internet of Things (IoT) and massive Machine Type Communications (mMTC), representing scenarios in which a very large number of devices are connected wirelessly (up to $300{,}000$ within a single cell~\cite{bockelmannmtc}). In a typical sensing and monitoring scenario, a large number of sensors are deployed across a large area. The sensors sporadically or periodically sense and transmit their observations to a destination node, e.g., in the cloud or at the edge of the cell, which processes the observations and possibly initiates some actions, such as adjusting control parameters, raising an alarm, etc. An example is the smart grid~\cite{jjnsmartmetering,gungorsmartgrid}, in which sensors continuously monitor the state of the power grid and report their observations to a central controller. Other examples include industrial manufacturing, where cameras are used to aid process tracking, and environmental monitoring, such as temperature and humidity~\cite{zanella14}.

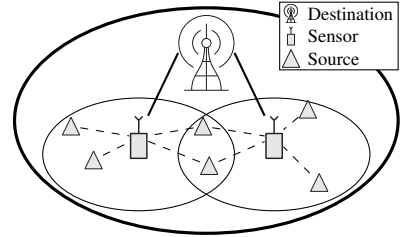
\begin{figure}
    \centering
    \begin{tikzpicture}[every node/.style={inner sep=0,outer sep=0},scale=0.75]
\pgfdeclareimage[width=0.75cm]{bs}{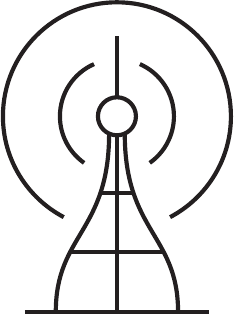}
\pgfdeclareimage[width=0.25cm]{source}{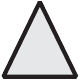}
\pgfdeclareimage[width=0.25cm]{sensor}{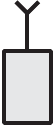}

\pgfdeclareimage[height=0.25cm]{bs_small}{clipart/Objects_tower_B.pdf}
\pgfdeclareimage[height=0.25cm]{source_small}{clipart/Objects_Iot_health.pdf}
\pgfdeclareimage[height=0.25cm]{sensor_small}{clipart/Objects_Iot_simple.pdf}

\draw[] (-1.2,-0.3) ellipse (1.7cm and 1.0cm);
\draw[] (1.2,-0.3) ellipse (1.7cm and 1.0cm);
\draw[very thick] (0,0.3) ellipse (3.4cm and 2cm);

\node (s11) at (-2.4,0.2) {\pgfuseimage{source}};
\node (s12) at (-2.0,-0.4) {\pgfuseimage{source}};
\node (s13) at (-0.06,0.2) {\pgfuseimage{source}};
\node (s21) at (0.05,-0.5) {\pgfuseimage{source}};
\node (s22) at (1.8,0.5) {\pgfuseimage{source}};
\node (s23) at (2.0,-0.8) {\pgfuseimage{source}};
\node (sensor1) at (-1.2,0.0) {\pgfuseimage{sensor}};
\node (sensor2) at (1.2,0.0) {\pgfuseimage{sensor}};

\node[] (bs) at (0,1.5) {\pgfuseimage{bs}};

\draw[dashed] (s11) -- (sensor1);
\draw[dashed] (s12) -- (sensor1);
\draw[dashed] (s13) -- (sensor1);
\draw[dashed] (s21) -- (sensor1);
\draw[dashed] (s13) -- (sensor2);
\draw[dashed] (s21) -- (sensor2);
\draw[dashed] (s22) -- (sensor2);
\draw[dashed] (s23) -- (sensor2);

\draw[thick] (sensor1.north east) -- (bs.west);
\draw[thick] (sensor2.north west) -- (bs.east);

\begin{scope}[shift={(-0.2,-0.5)}]
\draw[fill=white] (1.5,1.7) rectangle (3.6,2.9);
\node[] at (1.7,1.9) {\pgfuseimage{source_small}};
\node[] at (1.7,2.3) {\pgfuseimage{sensor_small}};
\node[] at (1.7,2.7) {\pgfuseimage{bs_small}};
\node[anchor=west] at (2.0,1.9) {\scriptsize{Source}};
\node[anchor=west] at (2.0,2.3) {\scriptsize{Sensor}};
\node[anchor=west] at (2.0,2.7) {\scriptsize{Destination}};
\end{scope}
\end{tikzpicture}
    \caption{The studied scenario comprising a number of sources that are observed by sensors, and a destination node, assumed to be located at the base station, that requests the observations from the sensors. The objective is to schedule the sensors so as to minimize the AoI of the sources.}
    \label{fig:sysmodel_illu}
\end{figure}

However, the large number of IoT devices imposes a significant constraint on the number of communicated observations. Hence, designing intelligent schemes for selecting when and which observations to transmit can provide notable gains in the performance of such systems. Although defining a data relevance or importance measure is inherently an application-specific task, for a large amount of IoT use cases the age of the observations is a reasonable criterion. Here the \emph{age} refers to the time elapsed since the generation of the most recent observation that is known to the destination. Returning to the previous examples, timely updates are critical in smart grid  systems, where action should be taken in case of anomaly. Similarly, environmental monitoring calls for frequent reporting. This wide applicability of using the age of the observations as a relevance measure has led to the notion of \emph{Age of Information} (AoI), which has inspired a large number of works (see~\cite{kosta17} and \cref{sec:relatedwork}).

Another characteristic that is shared among many IoT use cases is the fact that the sensor observations may be correlated. For instance, this may arise in sampling a physical phenomenon that is not confined only to the points where the sensors are deployed, but spans an area monitored by several sensors, see \cref{fig:sysmodel_illu}. Conversely, a single sensor may be able to observe several phenomena at the same time. Unless this correlation is considered in the design of the communication protocols, it is likely to lead to redundant data being communicated to the destination. In addition, the phenomena may be dynamic, with features that change  over time. For instance, the oberved entity may be mobile, entering/leaving the range of each individual sensor; or, a phenomenon may be observable by the sensors only in certain states, e.g., when  powered on.

We study a generalized version of the system depicted in \cref{fig:sysmodel_illu}, in which a destination node, located at the base station, monitors a set of dynamic sources with states defined by a Markov chain. The sources are monitored through a set of sensors. A sensor can observe a given source with some probability that depends on the state of the source. At each time slot, a scheduler at the destination node requests exactly one sensor to perform a measurement, and the sensor reports the measurement through an erasure channel to the destination node. This model can represent an industrial scenario with automatic guided vehicles (AGVs), observed by a number of cameras. The location of each AGV is represented by its state and the probability that a camera observes an AGV depends on whether the AGV is within the camera's field of view. We consider two cases: (1) the source states are fully observable to the scheduler, and (2) the source states are unobservable and need to be inferred through the measurements. In the case of AGVs, the fully observable case could represent the scenario in which the AGVs have localization systems so that their locations are known to the scheduler. Conversely, in the unobservable case, the AGV locations are only observed through the sampled cameras, e.g., using image recognition.

This scenario represents a somewhat extreme case of sensing, as sources are either perfectly observed or not observed at all. Yet,  it allows us to conveniently abstract detailed aspects, such as source coding and distortion.
Our overall aim is to design scheduling schemes that minimize the average AoI at the destination node, so that it has a timely overview of the sources. The combination of dynamic phenomena and the fact that sensors can observe multiple sources introduces a tradeoff between getting an update from a sensor that observes a single source for which the current data is with high AoI and likely outdated, or sampling a sensor that observes multiple sources, for which the current data is fresher (lower AoI). The long-term average AoI is minimized by a long-sighted strategy and depends on the dynamics of the source states, as illustrated through the following example.

\subsection{Motivating Example}
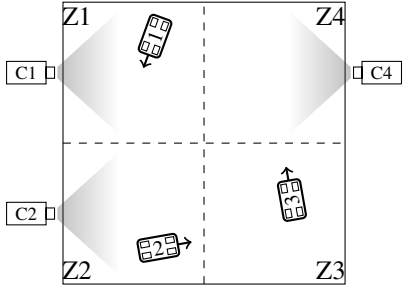
\begin{figure}
    \centering
    \begin{tikzpicture}[every node/.style={inner sep=0,outer sep=0},scale=0.75]

\draw (0,0) rectangle (5,5);
\draw[dashed] (2.5,0) -- ++(0,5);
\draw[dashed] (0,2.5) -- ++(5,0);
\node at (0.25,0.25) {$\text{Z2}$};
\node at (4.75,0.25) {$\text{Z3}$};
\node at (4.75,4.75) {$\text{Z4}$};
\node at (0.25,4.75) {$\text{Z1}$};

\begin{scope}[shift={(1.7,4.0)},rotate=70]
\draw [rounded corners=0.05cm,thick] (0.0,0.0) rectangle +(0.7,0.4) node[midway,rotate=70] {\footnotesize 1};
\draw (0.05,0.05) rectangle +(0.15,0.1);
\draw (0.5,0.05) rectangle +(0.15,0.1);
\draw (0.05,0.25) rectangle +(0.15,0.1);
\draw (0.5,0.25) rectangle +(0.15,0.1);
\draw[->,thick] (0.0,0.2) -- ++(-0.25,0.0);
\end{scope}
\begin{scope}[shift={(2.0,.9)},rotate=190]
\draw [rounded corners=0.05cm,thick] (0.0,0.0) rectangle +(0.7,0.4) node[midway,rotate=10] {\footnotesize 2};
\draw (0.05,0.05) rectangle +(0.15,0.1);
\draw (0.5,0.05) rectangle +(0.15,0.1);
\draw (0.05,0.25) rectangle +(0.15,0.1);
\draw (0.5,0.25) rectangle +(0.15,0.1);
\draw[->,thick] (0.0,0.2) -- ++(-0.25,0.0);
\end{scope}
\begin{scope}[shift={(3.8,1.8)},rotate=280]
\draw [rounded corners=0.05cm,thick] (0.0,0.0) rectangle +(0.7,0.4) node[midway,rotate=100] {\footnotesize 3};
\draw (0.05,0.05) rectangle +(0.15,0.1);
\draw (0.5,0.05) rectangle +(0.15,0.1);
\draw (0.05,0.25) rectangle +(0.15,0.1);
\draw (0.5,0.25) rectangle +(0.15,0.1);
\draw[->,thick] (0.0,0.2) -- ++(-0.25,0.0);
\end{scope}

\begin{scope}[shift={(6.0,3.55)}]
\draw[] (0,0) -- ++(-0.7,0) -- ++(0,0.3) -- ++(-0.15,0) -- ++(0,-0.2) -- ++(0.15,0) -- ++(0,0.3) -- ++(0.7,0) -- ++(0,-0.4);
\node at (-0.35,0.2) {\scriptsize C4};
\end{scope}
\draw[left color=white,right color=black!80,opacity=0.3,draw=none] (5.1,3.85) -- ++(-1.1,1.0) -- ++(0,-2.2) -- ++(1.1,1.0);

\begin{scope}[shift={(-1.0,3.95)},rotate=180]
\draw[] (0,0) -- ++(-0.7,0) -- ++(0,0.3) -- ++(-0.15,0) -- ++(0,-0.2) -- ++(0.15,0) -- ++(0,0.3) -- ++(0.7,0) -- ++(0,-0.4);
\node at (-0.35,0.2) {\scriptsize C1};
\end{scope}
\draw[left color=black!80,right color=white,opacity=0.3,draw=none] (-0.1,3.85) -- ++(1.1,1.0) -- ++(0,-2.2) -- ++(-1.1,1.0);

\begin{scope}[shift={(-1.0,1.45)},rotate=180]
\draw[] (0,0) -- ++(-0.7,0) -- ++(0,0.3) -- ++(-0.15,0) -- ++(0,-0.2) -- ++(0.15,0) -- ++(0,0.3) -- ++(0.7,0) -- ++(0,-0.4);
\node at (-0.35,0.2) {\scriptsize C2};
\end{scope}
\draw[left color=black!80,right color=white,opacity=0.3,draw=none] (-0.1,1.35) -- ++(1.1,1.0) -- ++(0,-2.2) -- ++(-1.1,1.0);

\end{tikzpicture}
    \caption{Example scenario of a factory with AGVs and three cameras that cover different zones. We consider the problem of scheduling sensors (cameras) to maximize the freshness of their observations for a set of sources (AGVs). Each sensor may observe multiple sources at the same time, e.g., if multiple AGVs are within the same zone.}
    \label{fig:motivation}
\end{figure}

Consider a factory with AGVs, where cameras are used to monitor the AGVs in order to ensure that they are operating normally. Suppose for simplicity that the factory is divided into four zones (zones 1--4), as illustrated in \cref{fig:motivation}. Zones 1, 2 and 4 are covered by cameras C1, C2 and C4, respectively, while zone 3 is hidden for the cameras. In each time slot, the AGVs move from one zone to the next in an anti-clockwise  direction, except in zone 1, where they spend two time slots before moving to the next zone, e.g., due to physical obstacles that prevents the AGVs from moving at their regular speed. Now, suppose further that due to capacity limitations it is only possible to request an image from one camera every time slot, and thus they need to be scheduled. Putting this into the framework presented earlier, the AGVs are sources with states that corresponds to their location, and the cameras are sensors. A na\"{i}ve scheduling scheme would be to schedule the cameras in a round-robin fashion. Although this may work reasonably well in a small and simplistic scenario like this, it will not perform well in more complex scenarios with many AGVs and zones. An alternative approach would be to exploit the structure of the factory and schedule the cameras only when the AGVs are expected to be inside the cameras fields of view. For instance, a myopic AoI scheduler would schedule the camera with the AGV that has been captured least recently once it enters a camera region. However, as illustrated in \cref{tab:motivation_aoi}, it is possible to do even better using a long-sighted policy that exploits the fact that the AGVs move slower in region 1, which is more likely to contain multiple AGVs in the same time slot.
With the optimal policy, the total age can be reduced from $38$ time slots in the myopic policy to $36$ time slots. This illustrates that finding the optimal scheduling policy is non-trivial, since it requires to account  both for the dynamics of the AGVs and the locations of all AGVs.

Note also that we have assumed deterministic movement of the AGVs, as well as ideal sensor observations. Scheduling becomes more challenging when the movement is random and the observations are non-ideal, since the scheduler may have to start scheduling an AGV several time slots before it is expected to enter the hidden zone. Similarly, it is more difficult to define a good policy when the locations of the AGVs are not fully observable, but revealed only when the AGVs are captured by one the cameras. Here the scheduler needs to consider both the AGV dynamics and its belief in the location estimates.

\begin{table*}[t]
\centering
\caption{Myopic and optimal policies for a single round in the scenario from \cref{fig:motivation} with four zones and three AGVs starting in zones Z1, Z2, and Z3 and with initial AoIs of 1, 1, and 4, respectively. The decisions are written $(\text{AoI 1},\text{AoI 2},\text{AoI 3})\to\text{decision}$. The Myopic policy achieves a total AoI of $38$ while the optimal achieves $36$.}\label{tab:motivation_aoi}
\begin{tabular}{l|llllllllll}
  \toprule
  Time slot & 1 & 2 & 3 & 4 & 5 & 6 & \multirow{4}{1.5em}{$\cdots$}\\
  Locations       & $(\text{Z1},\text{Z2},\text{Z3})$ & $(\text{Z2},\text{Z3},\text{Z4})$ & $(\text{Z3},\text{Z4},\text{Z1})$ & $(\text{Z4},\text{Z1},\text{Z1})$ & $(\text{Z1},\text{Z1},\text{Z2})$ & $(\text{Z1},\text{Z2},\text{Z3})$ & \\ 
  Myopic  & $(1,1,4)\to\text{C1}$ & $(1,2,5)\to\text{C4}$ & $(2,3,1)\to\text{C4}$ & $(3,1,2)\to\text{C4}$ & $(1,2,3)\to\text{C1}$ & $(1,1,4)\to\text{C1}$ & \\ %
  Optimal  & $(1,1,4)\to\text{C1}$ & $(1,2,5)\to\text{C4}$ & $(2,3,1)\to\text{C4}$ & $(3,1,2)\to\text{C1}$ & $(4,1,1)\to\text{C1}$ & $(1,1,2)\to\text{C1}$ & \\ %
   \bottomrule
\end{tabular}
\end{table*}

\subsection{Related Work}\label{sec:relatedwork}
The concept of Age of Information has recently received much attention in the context of Internet of Things (IoT), where it has been used to characterize the fundamental trade-off between update rate and timeliness in Poisson update systems a capacity constraint, see e.g.,~\cite{kaul12realtime,costa16,kosta17,sun17,najm19}. A common characteristic in many of these systems is that the AoI is a U-shaped function of the update rate and there exists a rate that minimizes the AoI.

The AoI literature for systems with multiple sources can roughly be divided into: (1) queuing-based systems, and (2) scheduling- or sampling-based systems. Notable contributions in the first category include~\cite{javani19} in which updates from each source are independently sensed and processed by multiple servers before being delivered to the destination. \cite{najm18,farazi19,moltafet20,kaul20} characterize the AoI when updates from multiple sources are delivered to the destination through a shared queue, and similar scenarios, but with scheduling and packet management, are studied in~\cite{yin18aoi,beytur19,maatouk20lex,moltafet20pm}. Packet management schemes for the case in which updates are queued at the source before being transmitted over a shared medium are studied in~\cite{kosta19}.

Contributions related to sampling of multiple sources include~\cite{tripathi17}, which proposes various scheduling policies to minimize the AoI of independent sources. \cite{bedewy20} considers the problem of scheduling sources to transmit through a shared queue. They define the problem as a Markov Decision Process (MDP) and derives the optimal policy. A similar problem is studied in~\cite{jiang19}, where each source has its own queue. In~\cite{tang20} sources are scheduled to transmit over erasure channels, whose erasure probabilities evolve according to a Markov chain. The authors present the optimal policy under the assumption that the scheduler can observe the instantaneous channel states before making a scheduling decision. In contrast to our work, they do not distinguish between sources and sensors.

Correlated sources and joint observations in the context of AoI have previously been studied in~\cite{jiang19randomfield,hribar19,he19,kalor19,shao20,zhou20}. In~\cite{jiang19randomfield}, the authors consider remote estimation of a Gaussian Markov random field that is sampled and queued by multiple sources before it is sent to the destination node. A similar situation with spatially correlated observations are studied in~\cite{hribar19}, where the authors show how the correlation can be exploited to reduce the energy consumption of the sources. For the case with discrete sources, \cite{he19} studies a scenario comprising a number of cameras that each captures a specific scene and transmits its images to a destination fog computing node over a shared channel. The same scenes may be captured by multiple cameras, and the authors propose a joint fog node assignment and scheduling scheme so that the average age of each scene at the fog nodes is minimized. We note that our work represents a generalization of that scenario by allowing each camera (sensor) to observe multiple scenes (sources), and also allows the sources to be dynamic. Similar scenarios have been studied in~\cite{kalor19,shao20,zhou20}. In \cite{kalor19} we proposed a scheduling mechanism for the case in which each sensor can observe multiple sources, and each source update is observed by a random subset of the sensors. This work differs from~\cite{kalor19} by assuming that observations are generated ``at will'', and by allowing sources to be stateful, so that the probability that a sensor observes a source depends on its state. Furthermore, we rigorously prove the optimality of the proposed policies. A similar scenario is considered in~\cite{shao20}, where multiple sensors observe a single source that generates updates and each update is observed by a given sensor with certain probability. The authors focus on the partially observable case where the sensor observations are unknown to the scheduler, and they study the performance of the myopic scheduling policy for the problem. Finally, \cite{zhou20} propose a scheduling scheme for the case with two types of correlated sensors, one receiving random updates from the source and one that is able to sample the source. They show that a long-sighted policy provides significantly lower AoI than the myopic policy.

The methods that we use in this paper to obtain scheduling policies have previously been used in the context of AoI. Due to their simplicity, stationary random policies have been studied extensively for other AoI problems in e.g.,~\cite{talak17multihop,talak20interference}, where the latter rigorously provides bounds on the optimality gap between an optimal random policies and minimal achievable AoI. However, because of the correlation in our model, the analysis is not directly transferable to our problem. The methodology used to obtain the optimal policy (i.e., solving an average-cost MDP) has been applied in e.g.,~\cite{hsu2017,bedewy20,tang20}, and our proofs share similarities with the ones in those works. Finally, myopic policies as approximate solutions to the MDP problem have been considered in various detail for different AoI problems in e.g.,~\cite{katoda18,tang20,jiang19}.

\subsection{Contributions and Paper Organization}
The ambition of this work is not to provide new fundamental results on AoI, but rather to study how dynamic sensors influence the scheduling decisions towards minimizing the AoI. The contributions can be summarized as follows.
\begin{itemize}
    \item We define a general system model that captures many IoT scenarios, such as camera monitoring, in which sources are dynamic and their observation probabilities change according to their state. Furthermore, sources and sensors have a many-to-many relationship: each source can be observed by multiple sensors and each sensor can observe multiple sources. We rigorously prove the existence of an optimal deterministic stationary policy for the problem when the source states are fully observable, and we show how it can be found using value iteration.
    \item We analyze three partially observable variants of the problem, in which the source states cannot be directly observed, but instead (1) the source states are revealed only for the observed sources after a delay, e.g., due to processing time, (2) the source identities, but not the states, are revealed of the observed sources, or (3) the sensor measurements do not reveal whether they contain a source. Furthermore, we devise heuristic scheduling policies for these problems.
    \item  We characterize  different policies numerically and study the impact of the source dynamics and the fact that sensors can observe multiple sources on the AoI. We show that the latter can be beneficial especially in the partially observable case, provided that the policy is able to account for the uncertainty of source states.
\end{itemize}

The remainder of the paper is organized as follows. In \cref{sec:sysmodel} we formally define the system model. We then analyze stationary random scheduling policies in \cref{sec:random_policy}, followed by the optimal scheduling policies in \cref{sec:policy_obs_states} for the case where the source states are fully observable.  \cref{sec:policy_unobs_states} treats two cases: (1) unobservable source states, where the sources can be identified from the observations, and (2) undetectable sources, where the observations do not reveal the detected sources to the scheduler. \Cref{sec:numres} presents numerical results, and finally the paper is concluded in \cref{sec:conclusion}.

\section{System Model and Problem Formulation}\label{sec:sysmodel}
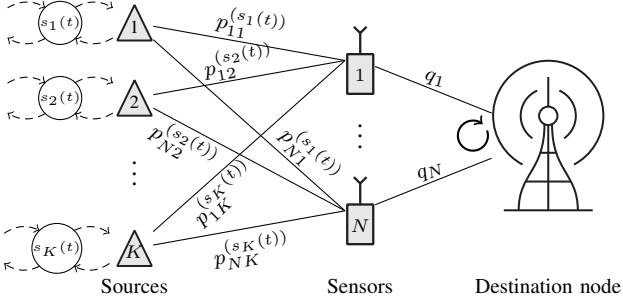
\begin{figure}
    \centering
    \begin{tikzpicture}[every node/.style={inner sep=0,outer sep=0}]
\pgfdeclareimage[width=1.5cm]{bs}{clipart/Objects_tower_B.pdf}
\pgfdeclareimage[width=0.5cm]{source}{clipart/Objects_Iot_health.pdf}
\pgfdeclareimage[width=0.4cm]{sensor}{clipart/Objects_Iot_simple.pdf}

\node (s11) at (0,1.5) {\pgfuseimage{source}};
\node at (0,1.45) {\scriptsize 1};
\node (s12) at (0,0.5) {\pgfuseimage{source}};
\node at (0,0.45) {\scriptsize 2};
\node (s21) at (0,-1.5) {\pgfuseimage{source}};
\node at (0,-1.55) {\scriptsize $K$};
\node at (0,-0.4) {$\vdots$};

\node (sensor1) at (3,1) {\pgfuseimage{sensor}};
\node at (3,0.8) {\scriptsize $1$};
\node (sensor2) at (3,-1) {\pgfuseimage{sensor}};
\node at (3,-1.2) {\scriptsize $N$};
\node at (3,0.1) {$\vdots$};

\node[] (bs) at (5.5,0) {\pgfuseimage{bs}};
\draw[thick,->] (4.5,0) +(20:0.2) arc(20:-310:0.2);

\node[circle,draw,minimum width=1.5em] (state1) [left=0.45 of s11] {\tiny $s_1(t)$};
\path[densely dashed,->] (state1) edge[bend left] ($(state1)+(-0.7,-0.2)$);
\path[densely dashed,->] ($(state1)+(-0.7,0.2)$) edge[bend left] (state1);
\path[densely dashed,->] (state1) edge[bend left] ($(state1)+(+0.7,0.2)$);
\path[densely dashed,->] ($(state1)+(+0.7,-0.2)$) edge[bend left] (state1);

\node[circle,draw,minimum width=1.5em] (state2) [left=0.45 of s12] {\tiny $s_2(t)$};
\path[densely dashed,->] (state2) edge[bend left] ($(state2)+(-0.7,-0.2)$);
\path[densely dashed,->] ($(state2)+(-0.7,0.2)$) edge[bend left] (state2);
\path[densely dashed,->] (state2) edge[bend left] ($(state2)+(+0.7,0.2)$);
\path[densely dashed,->] ($(state2)+(+0.7,-0.2)$) edge[bend left] (state2);

\node[circle,draw,minimum width=1.5em] (statek) [left=0.45 of s21] {\tiny $s_K(t)$};
\path[densely dashed,->] (statek) edge[bend left] ($(statek)+(-0.7,-0.2)$);
\path[densely dashed,->] ($(statek)+(-0.7,0.2)$) edge[bend left] (statek);
\path[densely dashed,->] (statek) edge[bend left] ($(statek)+(+0.7,0.2)$);
\path[densely dashed,->] ($(statek)+(+0.7,-0.2)$) edge[bend left] (statek);

\draw (s11) -- (sensor1.west) node[midway,above=0.05,rotate=-10] {\footnotesize $p_{11}^{(s_1(t))}$};
\draw (s12) -- (sensor1.west) node[pos=0.46,above=0.03,rotate=10] {\footnotesize $p_{12}^{(s_2(t))}$};
\draw (s21) -- (sensor1.west) node[pos=0.31,below=0.05,rotate=45] {\footnotesize $p_{1K}^{(s_K(t))}$};
\draw (s11) -- (sensor2.west) node[pos=0.76,above=0.05,rotate=-45] {\footnotesize $p_{N1}^{(s_1(t))}$};
\draw (s12) -- (sensor2.west) node[pos=0.2,below=0.05,rotate=-30] {\footnotesize $p_{N2}^{(s_2(t))}$};
\draw (s21) -- (sensor2.west) node[midway,below=0.1,rotate=10] {\footnotesize $p_{NK}^{(s_K(t))}$};

\draw (sensor1) -- (bs) node[midway,above=0.05,rotate=-15] {\footnotesize $q_1$};
\draw (sensor2) -- (bs) node[midway,above=0.05,rotate=20] {\footnotesize $q_N$};

\node[] at (0,-2) {\footnotesize{Sources}};
\node[] at (3,-2) {\footnotesize{Sensors}};
\node[align=center] at (5.5,-2) {\footnotesize{Destination node}};

\end{tikzpicture}
    \caption{A general system with $K$ stateful sources, $N$ sensors and one destination node. In each time slot $t$, a sensor is requested to perform a measurement. Sensor $n$ observes source $k$ with probability $p_{nk}^{(s_k(t))}$, where $s_k(t)$ denotes the state of source $k$. The measurement from sensor $n$ is delivered to the destination node with probability $q_n$.}
    \label{fig:sysmodel}
\end{figure}
We consider a system with $N$ sensors that observe $K$ stateful sources, as illustrated in \cref{fig:sysmodel}. The time is divided into time slots and in each slot a sensor, denoted $a_t\in\{1,2,\ldots,N\}$, is requested to perform a measurement, which is intended to be delivered to a central destination node, which also acts as the scheduler. The measurement contains a random subset of the sources, and the probability that a sensor measurement contains a source depends on the state of the source. We denote the state of source $k$ in time slot $t=1,2,\ldots$ by $s_{k}(t)\in\{1,2,\ldots,S_k\}$, and the probability that a measurement by sensor $n$ contains source $k$ in state $s$ by $p_{nk}^{(s)}$. $s_k(t)$ is modeled as a discrete-time Markov chain with transition matrix $\mathbf{R}_k=[r_{ij}^{(k)}]$, where $r_{ij}^{(k)}=\Pr(s_k(t)=j|s_k(t-1)=i)$. The Markov chain is assumed to be irreducible and aperiodic, and to have the stationary state distribution $\bm{\beta}_k$ defined as the vector that satisfies $\bm{\beta}_k\mathbf{R}_k=\bm{\beta}_k$. In the remainder of the paper, we will refer to the event that a sensor measurement contains a given source as the sensor \emph{observes} the source.
The measurements are delivered from the sensor to the destination node through an erasure channel with error probability $1-q_n$, which captures both errors in the request and the delivery transmissions.

This system model contains  interesting special cases, one of which is the capacity-constrained channel: in each time slot, a sensor can deliver updates from at most $C\le K$ sources. This can be modelled by defining, for each sensor, $\binom{K}{C}$ \emph{virtual sensors} that observe the various subsets of $C$ sources.

The overall goal is to minimize the average AoI at the destination node. Let $u_{n}(t)$ be a binary indicator of the event that sensor $n$ has been scheduled at time $t$, 
\begin{equation}
u_{n}(t)=\mathbbm{1}\left[a_t=n\right],
\end{equation}
where $\mathbbm{1}[\cdot]$ is the indicator function that equals one if the argument is true and otherwise is zero. The scheduled sensor can carry up to $K$ observations, and thus the AoI of source $k$ at the destination node, denoted by $\Delta_k(t)$, is described by:
\begin{equation}\label{eq:aoiprocess}
\Delta_k(t+1)=
\begin{cases}
  1 & \text{if }\sum_n\zeta_{nk}(t)u_{n}(t)=1,\\
  \Delta_k(t)+1 & \text{otherwise,}
\end{cases}
\end{equation}
where $\zeta_{nk}(t)$ is the Bernoulli random variable that indicates successful measurement and transmission of source $k$, i.e., is one with probability $p_{nk}^{(s_k(t))}q_n$ and zero with probability $1-p_{nk}^{(s_k(t))}q_n$. In other words, the AoI of source $k$ is reset if the sensor successfully measures source $k$ \emph{and} the measurement is successfully transferred to the destination. Each source is assumed to have a state in which it can be observed by a sensor, i.e., $\sum_{n=1}^N\sum_{s=1}^{S_k}q_n p_{nk}^{(s)}>0$ for all $k$. To simplify the subsequent notation we let $\Delta_k(0)=0$ for all $k$. Then, we denote the average AoI at the destination at time $t$ as
\begin{equation}\label{eq:deltaprime}
  \Delta'(t)= \frac{1}{K}\sum_{k=1}^K \Delta_k(t)
\end{equation}
minimize the long term average AoI
\begin{equation}
  \underset{\pi\in\Pi}{\text{minimize}} \limsup_{T\to\infty} \frac{1}{T} \E\left[\sum_{t=1}^T\Delta'(t)\right].
\end{equation}
by seeking a scheduling policy $\pi$. Here $\pi\in\Pi$ is a sequence of decision rules that maps the entire history of states and observations to the action space, potentially through a random mapping, and the expectation is taken over the decisions, the AoI process, and the source state transitions. The problem forms an average-cost MDP, in which the actions $a_t$ belong to the finite action space $\{1,2,\ldots,N\}$. The state space represents the current AoI of each source as well as the states of the sources. To formalize the setting, we denote the system state at the destination node as
\begin{equation}
    \bm{\Lambda}=(\bm{S},\bm{\Delta}),
\end{equation}
where $\bm{S}=\{s_1,s_2,\ldots,s_K\}$ and $\bm{\Delta}=\{\Delta_1,\Delta_2,\ldots,\Delta_K\}$ are the vectors of source states and AoIs, respectively. To be consistent with the MDP literature and to make the dependency on the action explicit, we define the cost of taking action $a$ in state $\bm{\Lambda}$ as the expected AoI in the following slot:
\begin{align}\label{eq:costfun}
    C(\bm{\Lambda},a)&=\frac{1}{K}\sum_{k=1}^K\Delta_k+1-\frac{1}{K}\sum_{n=1}^N\sum_{k=1}^K u_{n}p_{nk}^{(s_k)}q_n\Delta_k,
\end{align}
where we used  \cref{eq:aoiprocess}. Thus, optimization problem is:
\begin{equation}\label{eq:centralproblem}
  \underset{\pi\in\Pi}{\text{minimize}} \limsup_{T\to\infty} \frac{1}{T+1} \E\left[\sum_{t=0}^T C(\bm{\Lambda}_t,a_t)\bigg|\bm{\Lambda}_0\right].
\end{equation}

We assume that the system parameters are known. In  practice, they would need first to be learned (exploration phase); however, this learning phase is out of scope.  

\section{Random Sampling Policy}\label{sec:random_policy}
We start by studying stationary random sampling policies, under which the sensor selection does not depend on the state. Contrary to the adaptive policies presented later (i.e., with memory), the memoryless nature of the random sampling policy allows us to derive a closed-form expression of the expected long-term AoI under that policy.

\subsection{Long-term Average AoI}
The random sampling policy selects in each slot a sensor $a_t$ drawn from the distribution $\Pr(a_t=n)=\bar{a}_n$, where $\sum_{n=1}^N \bar{a}_n=1$, i.e., sensor $n$ is selected with probability $\bar{a}_n$. However, while the sampling decisions in each slot are independent, the observations are correlated over time as they depend on the source states, which do not change independently from one slot to another. To see this, consider the camera scenario considered in the introduction. Clearly, if a source is in zones 1, 3, or 4, it is likely that it will also be observable in the next time slot. Similarly, if on the other hand a source is in zone 2, it is likely to be invisible in the following time slot since zone 2 is not covered by any camera. In particular, the observations evolve according to a Markov chain which is independent of the scheduling decisions. Using this observation, we can construct the Markov chain of the AoI process and derive the following result.
\begin{theorem}\label{theo:randompolicyavgaoi}
The random sampling policy that schedules in each slot a sensor drawn according to the distribution $\Pr(a_t=n)=\bar{a}_n$ satisfying $\sum_{n=1}^N\sum_{s=1}^{S_k}\bar{a}_n q_n p_{nk}^{(s)}>0$ achieves an expected long-term average AoI of
\begin{equation}
    E[\Delta_{\text{random}}]=\frac{1}{K} \sum_{k=1}^K E[\Delta_{\text{random},k}],
\end{equation}
where $E[\Delta_{\text{random},k}]$ is the expected AoI of source $k$ given as
\begin{equation}
    E[\Delta_{\text{random},k}]=\bm{\beta}_k\mathbf{R}_k^{(\text{succ})}\left(\mathbf{I}-\mathbf{R}_k^{(\text{fail})}\right)^{-2}\mathbf{1}^T,
\end{equation}
$\mathbf{I}$ and $\mathbf{1}$ are the identity matrix and the row vector of all ones of appropriate dimensions, respectively, and
\begin{align}
    \mathbf{R}_k^{(\text{succ})}&=\diag\left(\mathbf{p}_k\right)\mathbf{R}_k,\\
    \mathbf{R}_k^{(\text{fail})}&=\left(\mathbf{I}-\diag\left(\mathbf{p}_k\right)\right)\mathbf{R}_k,\\
    \mathbf{p}_k&=\sum_{n=1}^N \bar{a}_n q_n\begin{bmatrix}p_{nk}^{(1)} & p_{nk}^{(2)} &\cdots & p_{nk}^{(S_k)}\end{bmatrix}.
\end{align}
\end{theorem}
\begin{proof}
See \cref{app:randompolicyavgaoiproof}.
\end{proof}

It follows immediately from \cref{theo:randompolicyavgaoi} that the expression for the average AoI can be greatly simplified when the sources have a single state, as presented in the following corollary.
\begin{corollary}\label{cor:randompolicyavgaoi_stateless}
For sources with a single state, i.e., $S_k=1$, the long-term average AoI of the random sampling policy is:
\begin{equation}
    E[\Delta_{\text{random}}]=\frac{1}{K} \sum_{k=1}^K \left(\frac{1}{\sum_{n=1}^N \bar{a}_n q_n p_{nk}^{(1)}}\right).
\end{equation}
\end{corollary}
\begin{proof}
When the sources have a single state the expressions reduce to scalar equations and
\ifdefined\SINGLECOL
\begin{align}
E[\Delta_{\text{random},k}]&=\left(\sum_{n=1}^N \bar{a}_n q_n p_{nk}^{(1)}\right)\times\left(1-\left(1-\sum_{n=1}^N \bar{a}_n q_n p_{nk}^{(1)}\right)\right)^{-2}\\
&=\left(\sum_{n=1}^N \bar{a}_n q_n p_{nk}^{(1)}\right)^{-1}.
\end{align}
\else
\begin{align}
E[\Delta_{\text{random},k}]&=\frac{\sum_{n=1}^N \bar{a}_n q_n p_{nk}^{(1)}}{\left(1-\left(1-\sum_{n=1}^N \bar{a}_n q_n p_{nk}^{(1)}\right)\right)^{2}}\\
&=\left(\sum_{n=1}^N \bar{a}_n q_n p_{nk}^{(1)}\right)^{-1}.
\end{align}
\fi
\noindent The result follows by averaging over the $K$ sources.
\end{proof}

\section{Optimal Policy with Observable States}\label{sec:policy_obs_states}
We now turn our attention to adaptive policies, i.e., policies that depend on the current system state, and we start by studying the case where the source states are observable to the scheduler. In this case, the system is fully observable and the problem can be analyzed as an average-cost MDP with unbounded costs (see e.g.,~\cite{ross83,puterman94,sennott89}). While the assumption of full observability is unrealistic in many practical scenarios, policies for this case will play a central role in the partially observable case that we consider later. We first prove that there exists a stationary policy that solves \cref{eq:centralproblem}, and we then show how such a policy can be found using relative value iteration.

\subsection{Structure of Optimal Policy}
Policies for MDPs can generally be characterized as randomized/deterministic and history dependent/stationary (Markovian)~\cite{puterman94}. Due to their simple structure, deterministic stationary policies are desired from both an analytical and practical perspective. A policy is said to be deterministic stationary if the same deterministic decision rule is used in each time slot, or, more formally, if $\bm{\Lambda}_{t_1}=\bm{\Lambda}_{t_2}$ implies that $a^*_{t_1}=a^*_{t_2}$, where $a^*_t$ is the optimal action at time $t$. Contrary to discounted MDPs, for which an optimal policy is guaranteed to be stationary, optimal policies for unbounded average-cost MDPs are in general not guaranteed to be deterministic and stationary, and may be both history dependent and stochastic~\cite{ross83,puterman94}.
One way to guarantee the existence of a deterministic stationary policy through a set of sufficient conditions provided in \cite{sennott89}, which we will show are satisfied for the problem in \cref{eq:centralproblem}. In the following paragraphs, we provide an overview of the conditions and the optimal policy, while we defer the formal proofs to \cref{app:proofoptstationarypolicy}.

The main idea behind the conditions is to show that a deterministic stationary policy for the corresponding $\alpha$-discounted problem exists in the limit as $\alpha\to 1$. To formalize, we define the value function as the total discounted cost under policy $\pi$ for a given discount factor $0<\alpha<1$ and initial state $\bm{\Lambda}$ as
\begin{equation}
    V_{\alpha,\pi}(\bm{\Lambda}) = \lim_{T\to\infty}\E\left[\sum_{t=0}^{T}\alpha^t C(\bm{\Lambda}_t,\alpha_t)\bigg|\bm{\Lambda}_0=\bm{\Lambda}\right],
\end{equation}
and let
\begin{equation}
    V_{\alpha}(\bm{\Lambda}) = \inf_{\pi}V_{\alpha,\pi}(\bm{\Lambda}).
\end{equation}
Note that $V_{\alpha}(\bm{\Lambda})\ge 0$ because all costs are positive.
The first condition, presented in \cref{app:proofoptstationarypolicy} as \cref{pro:condition1}, is that there exists a deterministic stationary policy that minimizes $V_{\alpha}(\bm{\Lambda})$. Note that this is generally not guaranteed due to the countable state space. Provided that it exists, the policy is given as
\begin{equation}
    a_t = \argmax_{a}\left\{ C(\bm{\Lambda}_t, a) + \alpha\E_{\bm{\Lambda}_{t+1}}\left[V_{\alpha}(\bm{\Lambda}_{t+1})|\bm{\Lambda}_t,a\right]\right\},
\end{equation}
and the value function satisfies the Bellman equation
\begin{equation}
    V_{\alpha}(\bm{\Lambda}) = \min_{a}\left\{ C(\bm{\Lambda}, a) + \alpha\E_{\bm{\Lambda}'}\left[V_{\alpha}(\bm{\Lambda}')|\bm{\Lambda},a\right]\right\}.
\end{equation}
A consequence of this is that the optimal discounted policy can be found using the iterative value iteration procedure~\cite{puterman94}
\begin{equation}
    V_{\alpha}^{n+1}(\bm{\Lambda}) = \min_{a}\left\{ C(\bm{\Lambda}, a) + \alpha\E_{\bm{\Lambda}'}\left[V_{\alpha}^{n}(\bm{\Lambda}')|\bm{\Lambda},a\right]\right\},\label{eq:valueiter}
\end{equation}
where $n=0,1,\ldots$ is the iteration number and $V_{\alpha}^{0}(\bm{\Lambda})$ is an arbitrary initial state. Then, $V_{\alpha}^{n+1}(\bm{\Lambda})\to V_{\alpha}(\bm{\Lambda})$ as $n\to\infty$ for any $\bm{\Lambda}$ and $0< \alpha < 1$.

Although the discounted problem converges to the optimal value function, we are ultimately interested in the undiscounted value function, i.e., with $\alpha=1$. However, \cref{eq:valueiter} diverges to $\infty$ for nonzero cost functions when $\alpha=1$. Instead, it is useful to consider the relative value function
$h_{\alpha}(\bm{\Lambda})=V_{\alpha}(\bm{\Lambda})-V_{\alpha}(\bm{\Lambda}_0)$, where $\bm{\Lambda}_0$ is an arbitrary fixed state. Using this, we define the relative value iteration procedure as
\begin{equation}
    h_{\alpha}^{n+1}(\bm{\Lambda}) = \min_{a}\left\{ C(\bm{\Lambda}, a) + \alpha\E_{\bm{\Lambda}'}\left[h_{\alpha}^{n}(\bm{\Lambda}')|\bm{\Lambda},a\right]\right\}-h_{\alpha}^{n}(\bm{\Lambda}_0),\label{eq:rel_valueiter}
\end{equation}
which can be interpreted as the total cost of state $\bm{\Lambda}$ relative to the cost of state $\bm{\Lambda}_0$. Contrary to the value iteration in \cref{eq:valueiter}, the relative value iteration procedure does not increase unboundedly for nonzero costs as $\alpha\to 1$ (provided that an optimal deterministic stationary policy exists), and it does not impact the sequence of maximizing actions~\cite{puterman94}.

The remaining two conditions required for an optimal deterministic stationary policy to exist relate to the boundedness of $h_{\alpha}(\bm{\Lambda})$. The second condition, presented in \cref{pro:condition2} of \cref{app:proofoptstationarypolicy}, states that $h_{\alpha}(\bm{\Lambda})$ must be uniformly bounded from below in both $\alpha$ and $\bm{\Lambda}$. The third condition concerns the upper bound of $h_{\alpha}(\bm{\Lambda})$. However, because the state space is countable and the cost function $C(\bm{\Lambda},a)$ is unbounded from above, $h_{\alpha}(\bm{\Lambda})$ is not uniformly upper bounded in $\bm{\Lambda}$. Instead, the third condition requires that $h_{\alpha}(\bm{\Lambda})$ is uniformly bounded from above only in $\alpha$, i.e., $h_{\alpha}(\bm{\Lambda})\le M_{\bm{\Lambda}}$ for all $\bm{\Lambda}$ and $\alpha$, \emph{and} that $\E_{\bm{\Lambda}_{t+1}}[M_{\bm{\Lambda}_{t+1}}|\bm{\Lambda}_t,a_t]< \infty$ for all $\bm{\Lambda}_t$ and $a_t$. This condition is stated in \cref{pro:condition3} in \cref{app:proofoptstationarypolicy}.

The fact that the conditions above are satisfied guarantees the existence of a deterministic stationary optimal policy. The result is summarized in the following theorem.
\begin{theorem}\label{theo:optstationarypolicy}
There exists a function $h(\bm{\Lambda})$ such that the policy
\begin{equation}
    a_t = \argmin_{a}\left\{ C(\bm{\Lambda}_t, a) + \E_{\bm{\Lambda}_{t+1}}\left[h(\bm{\Lambda}_{t+1})|\bm{\Lambda}_t,a\right]\right\}
\end{equation}
is optimal. Furthermore, for some constant $g$, $h(\bm{\Lambda})$ satisfies
\begin{equation}
    g + h(\bm{\Lambda}) = \min_{a}\left\{ C(\bm{\Lambda}, a) + \E_{\bm{\Lambda}'}\left[h(\bm{\Lambda}')|\bm{\Lambda},a\right]\right\}.\label{eq:reloptimalityeq}
\end{equation}
\end{theorem}
\begin{proof}
See \cref{app:proofoptstationarypolicy}.
\end{proof}

\subsection{Relative Value-Iteration}
As mentioned previously, the regular value iteration in \cref{eq:valueiter} diverges for $\alpha=1$, and so $h(\bm{\Lambda})$ is often computed directly using relative value iteration in \cref{eq:rel_valueiter}. However, relative value iteration is not directly applicable to our problem due to the infinite state space. Instead, we resort to an approximate solution and truncate the state space. This may be justified by the fact that some states are very unlikely to be reached, or because the AoI does not impose a higher penalty once it is above a certain threshold.
More specifically, we consider the $Q$-truncated AoI obtained by modifying the dynamics in \eqref{eq:aoiprocess} as
\begin{equation}
\Delta^{Q}_k(t+1)=
\begin{cases}
  1 & \text{if }\sum_n\zeta_{nk}(t)u_{n}(t)=1\\
  \min(Q, \Delta_k(t)+1) & \text{otherwise.}
\end{cases}
\end{equation}
In particular, because the unbounded system is stable and the cost function is nondecreasing in the AoI of each source, $\Lambda_k$, (see \cref{pro:nonnegative} in \cref{app:proofoptstationarypolicy}), we can achieve an arbitrary good approximation using a sufficiently large $Q$. However, a large value of $Q$ comes at the cost of an increased computational complexity of relative value iteration.
The value function for this problem, $h(\bm{\Lambda})$, can be found using relative value iteration~\cite{puterman94}, provided that the MDP is unichain, i.e., that every deterministic stationary policy has a single recurrent class plus a possibly empty class set of transient states. This is indeed the case for this problem, as can be seen by considering the state $\bm{\Lambda}=(\bm{S}_0,(Q,Q,\ldots,Q))$, which can be reached from all states with nonzero probability for any $\bm{S}_0$. The relative value iteration is summarized in \cref{lst:relvalueit}. 

It can be shown~\cite{bertsekas05_part2} that $h^{n}(\bm{\Lambda})$ satisfies $\underset{\bar{}}{\delta}^{n}\le \underset{\bar{}}{\delta}^{n+1}\le g \le \bar{\delta}^{n+1}\le \bar{\delta}^{n}$, where $g$ is the constant in \cref{eq:reloptimalityeq}, and
\begin{align}
    \underset{\bar{}}{\delta}^{n}&=\min_{\bm{\Lambda}} \{h^{n+1}(\bm{\Lambda})-h^{n}(\bm{\Lambda})\},\\
    \bar{\delta}^{n}&=\max_{\bm{\Lambda}} \{h^{n+1}(\bm{\Lambda})-h^{n}(\bm{\Lambda})\}.
\end{align}
Using this, an $\epsilon$-optimal approximation of the value function can be obtained by using the convergence criterion $\bar{\delta}^{n}-\underset{\bar{}}{\delta}^{n}< \epsilon$, which is guaranteed to be met in a finite number of iterations.

\begin{algorithm}
  \caption{Relative value iteration.}\label{lst:relvalueit}
\begin{algorithmic}
  \algrenewcommand\algorithmicfor{\textbf{For}}
  \algrenewtext{EndFor}{\textbf{End for}}
  \State Pick $\bm{\Lambda}_0$ arbitrarily
  \State $h^{0}(\bm{\Lambda})\gets 0,\quad\forall\bm{\Lambda}$
  \For{$n=0,1,\ldots$ until convergence}
  \State $\lambda^{n+1}\gets \min_{a}\left\{ C(\bm{\Lambda}_0, a) + \E_{\bm{\Lambda}'}\left[h^{n}(\bm{\Lambda}')|\bm{\Lambda}_0,a\right]\right\}$
  \ifdefined\SINGLECOL
  \State $h^{n+1}(\bm{\Lambda})\gets \min_{a}\left\{ C(\bm{\Lambda}, a) + \E_{\bm{\Lambda}'}\left[h^{n}(\bm{\Lambda}')|\bm{\Lambda},a\right]\right\} -\lambda^{n+1},\quad \forall \bm{\Lambda}$
  \else
  \State $h^{n+1}(\bm{\Lambda})\gets \min_{a}\left\{ C(\bm{\Lambda}, a) + \E_{\bm{\Lambda}'}\left[h^{n}(\bm{\Lambda}')|\bm{\Lambda},a\right]\right\}$
  \State $\qquad\qquad\qquad\qquad -\lambda^{n+1},\quad \forall \bm{\Lambda}$
  \fi
  \EndFor
\end{algorithmic}
\end{algorithm}

\subsection{Approximate Policies}
Each epoch in the value iteration algorithm requires a pass over all states, actions, and possible future states. With a state space of size $S_1\times S_2\times\ldots\times S_K\times Q^K$, this is infeasible even for small scenarios.
The problem of large dimensionality has been addressed in several ways in the literature, most notably by computing the value function for only a subset of the states, and then use those values to fit an approximate function using e.g., a decision tree or a neural network, or by using reinforcement learning methods to approximate the optimal policy in some high dimensional feature space~\cite{bertsekas05_part2, sutton18}. While such methods often perform well in practice, they are typically very problem dependent and their convergence properties are not well understood. Another strategy that has frequently been applied to AoI problems is to frame the problem as a restless multi-armed bandit problem (e.g., in~\cite{hsu2017}), for which good heuristic methods exist, such as the Whittle index policy~\cite{whittle90}. To frame the problem as a restless multi-armed, the problem must be separable into independent sub-problems, such as the scheduling of individual sources or sensors. Since the sensors in the system that we consider can observe multiple sources and each source can be observed by multiple sensors, such a separation is not possible.

Instead, as we are mainly concerned with how the common observation model impacts the AoI, we resort to the myopic policy when the state space is too large to allow the value function to be computed. The myopic policy schedules the sensors with the lowest immediate cost, i.e.,
\begin{equation}\label{eq:myopic}
    a_t = \argmin_{a} \left\{C(\bm{\Lambda}_t, a)\right\}.
\end{equation}
The myopic policy is generally sub-optimal, and as we illustrate in \cref{sec:numres} the memory in the AoI process makes the myopic policy sub-optimal even for very simple scenarios, such as one in which the sources only have a single state, i.e., $S_k=1$ for all $k$. Nevertheless, as presented in the following proposition, it turns out that the myopic policy is optimal when all of the following two conditions are satisfied: (1) $S_k=1$ for all $k$, (2) each sensor observes exactly one source with nonzero probability $p$, i.e., for all sensors $n=1,2,\ldots,N$, $p_{nk}^{(1)}=p$ for some $k=k_n'$ and $p_{nk}^{(1)}=0$ for $k\neq k_n'$.

\begin{proposition}\label{prop:myopicoptimal}
    The myopic policy
    \begin{equation}
        a_t = \argmin_{a}\left\{ C(\bm{\Lambda}_t, a)\right\}
    \end{equation}
    is optimal if the sources have a single state, i.e., $S_k=1$ for all $k$, each sensor observes exactly one source with probability $p$ and all other sources with probability $0$, and the channel error probabilities are equal, i.e., $q_n=q$ for all $n$.
\end{proposition}
\begin{proof}
This is a special case of~\cite[Theorem 5]{katoda18} with weight $a=1$.
\end{proof}

\section{Policies for Unobservable States}\label{sec:policy_unobs_states}
In this section, we consider the problem where part of the system states are unobservable and need to be inferred through the measurements.
We consider two extreme instances of the problem. First, we consider the case where the sources, but not their states, can be detected from a measurement with a fixed delay of $\tau\ge 0$ time slots. This delay can, for example, represent communication delay and processing time of a cloud-based image recognition process in the camera scenario. Such processing is typically much slower than the scheduling interval (i.e., scheduling can happen every few milliseconds whereas cloud processing can be in the order of tens to hundreds of milliseconds). We also discuss how to handle the situation in which the source state is revealed by the measurements, as in the camera scenario when the states represent the physical location of the sources.
In the second instance, we assume that the sources cannot be detected based on the sensor measurement, which is for instance the case for a camera without image recognition, or touch sensors, where several sources can press the sensor, but the sensor reading does not reveal the specific source. In contrast to the first instance where only the source states are unobservable, in this case the source AoIs are also unobservable.
In both cases, we assume that the states and transitions probabilities are known. Sometimes, this information may not be available and will need to be learned using e.g., reinforcement learning methods. However, such approaches are beyond the scope of this paper.

The unobservable states lead to a partially observable Markov decision process (POMDP), in which the scheduler needs to keep track of the information obtained through the measurements. This information can be represented as a vector that contains all the previous observations and actions, and the instances at which they were taken/observed. By treating the information vector as part of the state, the POMDP becomes a fully observable MDP. However, the dimension of the information vector increases in every time slot, which makes the problem difficult to solve. It turns out that a probability distribution over the states, referred to as the belief states, $b(\bm{\Lambda})$, is a sufficient statistic~\cite{bertsekas05_part1}. By including the belief state in the system state and redefining the expected cost as
\begin{equation}\label{eq:pomdpcost}
  \bar{C}(\bm{\Lambda},a)=\sum_{\bm{\Lambda}}b(\bm{\Lambda})C(\bm{\Lambda},a),
\end{equation}
we can obtain an MDP that can be solved using relative value iteration as outlined in \cref{sec:policy_obs_states} (provided that a stationary policy exists). Alternatively, the cost function can be incorporated directly into the myopic policy in~\eqref{eq:myopic} if the size of the state space renders value iteration intractable. In the remainder of this section, we first derive the belief state of the two considered instances, and then conclude the section by discussing approximate solutions to the POMDP, which are usually necessary due to the state space expansion caused by the inclusion of the belief vector.
In the following, we will make use of the fact that the source AoIs and states are conditionally independent, and factor the belief vector into beliefs for each source, $\Lambda_k=(s_k,\Delta_k)$, as
\begin{equation}
  b(\bm{\Lambda})=\prod_{k=1}^K b_{k}(\Lambda_k).
\end{equation}

\subsection{Detectable Sources with Delay}
For the case when the sources are detectable from the measurements with a delay $\tau\ge 0$, we define the observation at time $t$ as $\bm{O}^{(t)}=(\mathcal{Q}^{(t)},\zeta^{(t-\tau)}_{1},\ldots,\zeta^{(t-\tau)}_{K})$, where $\mathcal{Q}^{(t)}=1$ indicates successful transmission from the source to the destination at time $t$, otherwise $\mathcal{Q}^{(t)}=0$, and $\zeta^{(t-\tau)}_{k}=1$ if source $k$ was observed at time $t-\tau$, otherwise zero. Note that $\zeta^{(t-\tau)}_{1},\ldots,\zeta^{(t-\tau)}_{K}$ are delayed observations while $\mathcal{Q}^{(t)}$ is instantaneous.

Due to the delayed observations, we maintain in the beginning of each time slot $t$ the belief of the state in the beginning of time $t-\tau$, which we denote by $\tilde{b}_{k}^{(t)}(\Lambda_k^{(t-\tau)})$, using the observations received up until and including time $t-1$. To predict the system state belief at time $t$, denoted by $b_{k}^{(t)}(\Lambda_k^{(t)})$ and needed for the scheduling decision, we then evolve the state belief from time $t-\tau$ to $t$ using the state dynamics and the transmission outcomes $\mathcal{Q}^{(t-\tau)},\ldots,\mathcal{Q}^{(t-1)}$, which are all known at time $t$.

We start by deriving the belief update rule for $\tilde{b}_{k}^{(t)}(\Lambda_k^{(t-\tau)})$. Because we know at time $t$ the sources that were observed at times $1,2,\ldots,t-\tau-1$, we know exactly their AoI $\Delta_k$ at time $t-\tau$. On the other hand, we have uncertainty about the source states $s_k$, which are not directly observed, but need to be inferred through the observations. Thus, to simplify the notation we limit the focus to the source states and assume that the belief is only evaluated at the actual AoI, in which case the belief of the state can be expressed as
\begin{align}
  &\tilde{b}^{(t)}_k\left(s_k^{(t-\tau)}\right)\nonumber\\
  &\quad=\Pr\left(s_k^{(t-\tau)}|\bm{O}^{(t-\tau-1:t-1)},a^{(t-\tau-1)},\tilde{b}^{(t-1)}_{k}\right)\nonumber\\
    &\quad=\sum_{s_k^{(t-\tau-1)}=1}^{S_k}\Pr\left(s_k^{(t-\tau)}|s_k^{(t-\tau-1)}\right)\nonumber\\
      &\quad\qquad\times\Pr\left(s_k^{(t-\tau-1)}|\bm{O}^{(t-\tau-1:t-1)},a^{(t-\tau-1)},\tilde{b}^{(t-1)}_{k}\right)\nonumber\\
      &\quad\propto\sum_{s_k^{(t-\tau-1)}=1}^{S_k}\Pr\left(s_k^{(t-\tau)}|s_k^{(t-\tau-1)}\right)\nonumber\\
      &\quad\qquad\times\Pr\left(\mathcal{Q}^{(t-\tau-1)},\zeta_{k}^{(t-\tau-1)}|a^{(t-\tau-1)},s_k^{(t-\tau-1)}\right)\nonumber\\
      &\quad\qquad\times \tilde{b}^{(t-1)}_k\left(s_k^{(t-\tau-1)}\right),\label{eq:b_t_evolution}
\end{align}
where the last step follows from Bayes' theorem and we used the notation $x^{(a:b)}=(x^{(a)},x^{(a+1)},\ldots,x^{(b)})$. The transition probabilities are given as $\Pr(s_k^{(t-\tau)}|s_k^{(t-\tau-1)})=r_{ij}^{(k)}$ for $s_k^{(t-\tau)}=j$ and $s_k^{(t-\tau-1)}=i$, and the observation probabilities can be computed as
\ifdefined\SINGLECOL
\begin{align}
  \Pr\left(\mathcal{Q},\zeta_{k}|a,s_k\right)&=
    \begin{cases}
      q_{a}\Big[ \zeta_{ak}p_{ak}^{(s_k)}+(1-\zeta_{k})(1-p_{ak}^{(s_k)})\Big] & \text{if } \mathcal{Q}=1,\\
      1-q_{a} & \text{otherwise,}
    \end{cases}
\end{align}
\else
\begin{align}
  \Pr\left(\mathcal{Q},\zeta_{k}|s_k,a\right)&=
    \begin{cases}
      \!\begin{aligned}
      &q_{a}\Big[ \zeta_{ak}p_{ak}^{(s_k)}\\
      &\quad+(1-\zeta_{ak})(1-p_{ak}^{(s_k)})\Big]
      \end{aligned} & \text{if } \mathcal{Q}=1,\\
      1-q_{a} & \text{otherwise,}
    \end{cases}
\end{align}
\fi
where we omitted the time indices for clarity.

Using this belief, we can compute the belief of the state in the beginning of time $t$ using observations $\bm{O}^{(t-\tau:t-1)}$, which we will denote by $b_{k}^{(t)}(\Lambda_k^{(t)})$. Note that if $\tau=0$ this distribution reduces to $b_{k}^{(t)}(\Lambda_k^{(t)})=\tilde{b}_{k}^{(t)}(\Lambda_k^{(t)})$. For $\tau > 0$, we can derive $b_{k}^{(t)}(\Lambda_k^{(t)})$ by exploiting the Markovian structure of the problem and factorize the distribution as
\begin{align}
    b_{k}^{(t)}\left(\Lambda_k^{(t)}\right)&=
    \sum_{\Lambda_{k}^{(t-1)}}\cdots \sum_{\Lambda_{k}^{(t-\tau)}}
    \tilde{b}^{(t)}_{k}\left(\Lambda_{k}^{(t-\tau)}\right)\nonumber\\
    &\qquad\prod_{i=1}^{\tau}\Pr\left(\Lambda_k^{(t-i+1)}|\mathcal{Q}^{(t-i)},a^{(t-i)},\Lambda_k^{(t-i)}\right)\nonumber\\
  \begin{split}
    &\propto
    \sum_{\Lambda_{k}^{(t-1)}}\cdots \sum_{\Lambda_{k}^{(t-\tau)}}
    \tilde{b}^{(t)}_{k}\left(\Lambda_{k}^{(t-\tau)}\right)\\
    &\qquad\prod_{i=1}^{\tau} \Pr\left(\mathcal{Q}^{(t-i)}|a^{(t-i)},\Lambda_k^{(t-i)}\right)\\
    &\qquad\quad\times\Pr\left(\Lambda_{k}^{(t-i+1)}|a^{(t-i)},\Lambda_{k}^{(t-i)}\right).
  \end{split}
\end{align}
Here, $\Pr(\mathcal{Q}^{(t-i)}=1|a^{(t-i)}\Lambda_k^{(t-i)})=q_{a^{(t-i)}}$ and $\Pr(\mathcal{Q}^{(t-i)}=0|a^{(t-i)}\Lambda_k^{(t-i)})=1-q_{a^{(t-i)}}$, and
\begin{align}
  \Pr(\Lambda_k^{(t-i+1)}|a^{(t-i)},\Lambda_k^{(t-i)})=q_{a^{(t-i)}} p_{a^{(t-i)}k}^{(s)}r_{ss'}^{(k)}
\end{align}
when $s_k^{(t-i+1)}=s'$, $s_k^{(t-i)}=s$, and $\Delta_k^{(t-i+1)}=1$;
\begin{align}
  \Pr(\Lambda_k^{(t-i+1)}|a^{(t-i)},\Lambda_k^{(t-i)})=(1-q_{a^{(t-i)}}p_{a^{(t-i)}k}^{(s)})r_{ss'}^{(k)}
\end{align}
when $s_k^{(t-i+1)}=s'$, $s_k^{(t-i)}=s$, $\Delta_k^{(t-i+1)}=\Delta_k^{(t-i)}+1$; and otherwise $\Pr(\Lambda_k^{(t-i+1)}|a^{(t-i)},\Lambda_k^{(t-i)})=0$.

Note that in some applications it may be natural to assume that a source observation reveals the true state of a source. This scenario can be captured by using the alternative update rule $\tilde{b}_t^k(s_k^{(t-\tau)})=\Pr(s_k^{(t-\tau)}|s_k^{(t-\tau-1)})$ when source $k$ is observed (and the one in \cref{eq:b_t_evolution} when source $k$ is not observed), where $s_k^{(t-\tau-1)}$ is the true state of source $k$ revealed at time $t-1$.

\cref{fig:beliefexample} shows an example of how the belief evolves for $\tau=0$ for a single source that moves around in a grid. The solid red circle denotes  the source location (state). In each time slot one of the sensors is scheduled (empty circle), chosen from the ones that can observe a single state with probability $1$. Initially, at $t=0$, the belief is uniform over all source states. As time passes, the belief changes as indicated at $t=t_0$, where the scheduled sensor observes the source. The  belief is updated at  $t=t_0+1$, reflecting the knowledge obtained at $t=t_0$, and again at $t=t_0+2$.

\begin{figure}
    \centering
    \begin{tikzpicture}[tight background]
\begin{groupplot}[group style={
\ifdefined\SINGLECOL
group size= 4 by 1,
horizontal sep=3cm,
\else
group size= 4 by 1,
horizontal sep=0.5cm,
\fi
ylabels at=edge left,
vertical sep=0.75cm,
every plot/.style={%
every axis title shift=0pt,
colormap name=viridis
}},
ymin=-0.5,
ymax=7.5,
xmin=-0.5,
xmax=7.5,
height=3.25cm,width=3.25cm,
scaled y ticks=false,
yticklabel style={
  /pgf/number format/fixed,
  /pgf/number format/precision=3
},
try min ticks=2,
]
\ifdefined\SINGLECOL
\newcommand*{\varxshift}{-2.5cm}
\else
\newcommand*{\varxshift}{0}
\fi
\nextgroupplot[title={$t=0$},font=\footnotesize,point meta min=0.015,point meta max=0.017]
\addplot[matrix plot,point meta=explicit,mesh/cols=8,mesh/rows=8] table [col sep=comma,meta index=2] {data/belief_example_0.csv};
\draw[color=red,fill] (axis cs:0,5) circle[radius=0.3];
\draw[color=red,ultra thick] (axis cs:0,3) circle[radius=0.5];
\nextgroupplot[title={$t=t_0$},font=\footnotesize,point meta min=0.015,point meta max=0.017,xshift=\varxshift]
\addplot[matrix plot,point meta=explicit,mesh/cols=8,mesh/rows=8] table [col sep=comma,meta index=2] {data/belief_example_1.csv};
\draw[color=red,fill] (axis cs:3,2) circle[radius=0.3];
\draw[color=red,ultra thick] (axis cs:3,2) circle[radius=0.5];
\nextgroupplot[title={$t=t_0+1$},font=\footnotesize,point meta min=0.0,point meta max=0.2]
\addplot[matrix plot,point meta=explicit,mesh/cols=8,mesh/rows=8] table [col sep=comma,meta index=2] {data/belief_example_2.csv};
\draw[color=red,fill] (axis cs:3,3) circle[radius=0.3];
\draw[color=red,ultra thick] (axis cs:3,1) circle[radius=0.5];
\nextgroupplot[title={$t=t_0+2$},font=\footnotesize,point meta min=0.0,point meta max=0.2,xshift=\varxshift]
\addplot[matrix plot,point meta=explicit,mesh/cols=8,mesh/rows=8] table [col sep=comma,meta index=2] {data/belief_example_3.csv};
\draw[color=red,fill] (axis cs:3,3) circle[radius=0.3];
\draw[color=red,ultra thick] (axis cs:3,2) circle[radius=0.5];

\end{groupplot}

\end{tikzpicture}
    \caption{Example of how the (normalized) belief evolves for $\tau=0$ for a single source that moves around in a grid as various sensors (squares) are scheduled.}%
    \label{fig:beliefexample}
\end{figure}
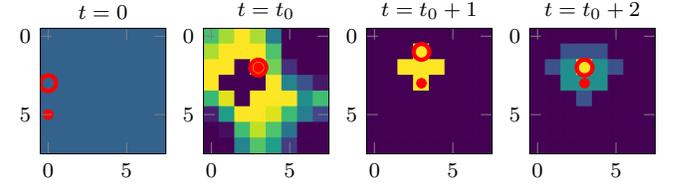

\subsection{Undetectable Sources}
In some cases, the scheduler may not be able to observe the sources from the sensor readings, e.g., if they are complex to extract or if the sensor data is encrypted. When sensor measurements reveal nothing about the sources to the scheduler, the source state belief remains constant, but the scheduler needs instead to keep track of its belief over the AoIs. Since the source states are unobservable, we assume they are in stationarity, and denote by $\beta_k^{(s_k)}$ the steady-state probability that source $k$ is in state $s_k$. In this regime we can assume that each source has a single state with observation probabilities
\begin{equation}
    \hat{p}_{nk}=\sum_{s=1}^{S_k} \beta_k^{(s)} p_{nk}^{(s)}.
\end{equation}
The AoI belief can then be constructed as follows.
\begin{equation}
    b(\bm{\Delta})=\prod_{k=1}^{K}b^k(\Delta_k),
\end{equation}
with update rule
\begin{align}
    b_t^k(\Delta_k')&=\Pr(\Delta_k'|\bm{O},a,b_{t-1}^k)\nonumber\\
    &=\begin{cases}    
    b_{t-1}^k(\Delta_k'-1)(1-q\hat{p}_{ak}) & \text{for } \Delta_k' > 1,\\
    q\hat{p}_{ak}& \text{for } \Delta_k' = 1.
    \end{cases}
\end{align}

In addition to the two considered cases, there is a third case in which objects are detectable, but they are all identical so only the number of observed sources is revealed. We omit this case, as it is  straightforward to derive the belief update rule.

\subsection{Approximate POMDP Policies}
Because $b(\bm{\Lambda})$ is constructed from a countable MDP, the state space of the POMDP is also countable. Unfortunately, the conditions used in \cref{sec:policy_obs_states} to prove the existence of a optimal stationary policy are not straightforward to apply here since there is no natural ordering of the belief states. Furthermore, alternative sufficient conditions such as the ones given in~\cite{fernandez91, arapostathis93} are difficult to prove for the general source state model that we consider. Finally, relative value iteration is challenging due to the state augmentation caused by the belief vector.
As a result, one often has to resort to sub-optimal policies based on the value function obtained in \cref{sec:policy_obs_states} for the observable state space using the two methods outlined below. Note that both methods can be applied in a myopic fashion as well by discarding the expectation term $\E_{\bm{\Lambda}_{t+1}}[\cdot]$.

\subsubsection{Maximum Likelihood (ML) Policy}
One way to apply the value functions for the fully observable MDP to the POMDP is to extract the maximum likelihood (ML) state from the belief states, and act as if it was the true state~\cite{cassandra98}. Thus, the partial observability is hidden from the agent, which picks the action
\begin{equation}
    a_t = \argmin_{a}\left\{ C(\bm{\Lambda}_t^{\text{ML}}, a) + \E_{\bm{\Lambda}_{t+1}}\left[h(\bm{\Lambda}_{t+1})|\bm{\Lambda}_t^{\text{ML}},a\right]\right\},
\end{equation}
where $\bm{\Lambda}_t^{\text{ML}}=\argmax_{\bm{\Lambda}}b_t(\bm{\Lambda})$.
The policy has two major drawbacks. The first is that it will not take actions to gain information about the states, which in some scenarios can lead to situations in which the belief state will never change, and thus the agent may end up in some sub-optimal state~\cite{littman95}. The second drawback is that the policy discards a large part of the information contained in the belief state, and acts as if there was no uncertainty. As a result, if all the states are almost equally likely, but the action that minimizes the value for the most likely state results in a significant increase in the costs in the other states, following the maximum likelihood policy may result in a very high cost. This situation may arise in the scenario from \cref{fig:motivation} if the most likely state is that all AGVs are in zone Z3, which is not covered by any of the cameras. In this state the costs incurred by the available actions are indifferent, and the agent may decide to schedule zone Z1 even though the AGVs are probably more likely to be in the zones next to Z3, i.e., Z4 or Z2.

\subsubsection{Q-MDP Policy}
To improve the second drawback of the maximum likelihood policy, the Q-MDP policy picks the action that minimizes the expected cost~\cite{littman95}. More specifically, the agent selects the action according to the rule
\ifdefined\SINGLECOL
\begin{align}
  a_t = \argmin_{a}\Bigg\{ \sum_{\bm{\Lambda}_t}b_t(\bm{\Lambda}_t) \left( C(\bm{\Lambda}_t, a)+\E_{\bm{\Lambda}_{t+1}}\left[h(\bm{\Lambda}_{t+1})|\bm{\Lambda}_t,a\right]\right)\Bigg\}.
\end{align}
\else
\begin{align}
  \begin{split}
    a_t &= \argmin_{a}\Bigg\{ \sum_{\bm{\Lambda}_t}b_t(\bm{\Lambda}_t)\\
    &\qquad\times \left( C(\bm{\Lambda}_t, a)+\E_{\bm{\Lambda}_{t+1}}\left[h(\bm{\Lambda}_{t+1})|\bm{\Lambda}_t,a\right]\right)\Bigg\}.
  \end{split}
\end{align}
\fi
Although the Q-MDP policy addresses the second drawback of the ML policy, it still does not favor information gaining actions. The advantage of Q-MDP comes at the cost of being computationally more demanding than the ML policy, but it can be efficiently approximated by Monte Carlo estimation by drawing state samples from the belief distribution and averaging their costs.

\section{Numerical Results}\label{sec:numres}
In this section, we apply the methods presented in the previous sections to a number of scenarios. We first study the optimal and myopic policies for a number of small toy scenarios, which highlight both the impact of the source dynamics and the possibility for sensors to observe multiple sources. We then consider more complex scenarios and investigate the importance of the amount of state information at the scheduler by comparing the AoI of the different policies. The simulation results in the section represent averages over 10 runs of $100\,000$ time slots, where the initial $10\,000$ time slots from each run have been discarded to minimize the impact of the initial state. In all scenarios with relative value iteration, we have truncated the AoI at a value, $Q$, as high as possible, while keeping the computational time reasonable and ensuring that the probability of having an AoI greater than $Q$ is negligible. Thus, the results are likely to be very close to the ones that would have been achieved with unbounded AoI.

\subsection{Toy Scenarios}
We consider the three toy scenarios illustrated in \cref{fig:toyscenario}, and assume that the source states are fully observable.
The first scenario, shown in \cref{fig:toyschemea}, contains two stateless sources, which are observed by three sensors. Sensors 1 and 2 observe only sources 1 and 2, respectively, with probability $p$, while sensor 3 observes both source 1 and 2 with probability $1-p$. Thus, the scenario reveals when correlated observations can be beneficial: Intuitively, one would expect that it is beneficial to primarily schedule sensor 3 when $p$ is small, and to schedule sensors 1 and 2 when $p$ is large. The scenario in \cref{fig:toyschemeb} is intended to show the impact of the source states. Each of the two sources has two states, and the sensors can observe the sources only in one of the states. Furthermore, to show how the source observation probabilities and the source state probabilities impact the scheduling decisions, the source transition probabilities are reversed for the two sources, and sensor 1 observes source 1 with probability $p$, while source 2 is observed by sensor 2 with probability $1-p$. Finally, the third scenario in \cref{fig:toyschemec} combines the first two by including both source states and a sensor that can observe both sources.

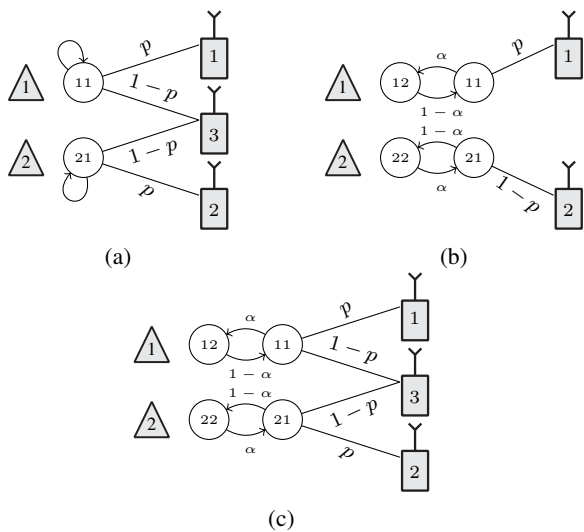
\begin{figure}
    \centering
    \begin{subfigure}[t]{%
    \ifdefined\SINGLECOL
    0.3\linewidth
    \else
    0.49\linewidth
    \fi
}
\begin{center}
    \begin{tikzpicture}[every node/.style={inner sep=0,outer sep=0}]
    \pgfdeclareimage[width=1.5cm]{bs}{clipart/Objects_tower_B.pdf}
    \pgfdeclareimage[width=0.5cm]{source}{clipart/Objects_Iot_health.pdf}
    \pgfdeclareimage[width=0.4cm]{sensor}{clipart/Objects_Iot_simple.pdf}
    
    \node (s11) at (1.0,1.5) {\pgfuseimage{source}};
    \node at (1.0,1.45) {\scriptsize 1};
    \node (s12) at (1.0,0.5) {\pgfuseimage{source}};
    \node at (1.0,0.45) {\scriptsize 2};
    
    \node (sensor1) at (3.5,2) {\pgfuseimage{sensor}};
    \node at (3.5,1.85) {\scriptsize $1$};
    \node (sensor2) at (3.5,-0) {\pgfuseimage{sensor}};
    \node at (3.5,-0.2) {\scriptsize $2$};
    \node (sensor3) at (3.5,1.0) {\pgfuseimage{sensor}};
    \node at (3.5,0.8) {\scriptsize $3$};
    
    \node[circle,draw,minimum width=1.5em] (state11) [right=0.25 of s11] {\tiny $11$};
    \node[circle,draw,minimum width=1.5em] (state21) [right=0.25 of s12] {\tiny $21$};
    
    \path[->] (state11) edge[out=140,in=90,looseness=6] (state11);
    \path[->] (state21) edge[out=280,in=230,looseness=6] (state21);
    
    \draw (state11) -- (sensor1.west) node[midway,above=0.1,rotate=35] {\footnotesize $p$};
    \draw (state21) -- (sensor2.west) node[midway,below=0.1,rotate=-35] {\footnotesize $p$};
    \draw (state11) -- (sensor3.west) node[midway,above=0.1,rotate=-16] {\footnotesize $1-p$};
    \draw (state21) -- (sensor3.west) node[midway,below=0.1,rotate=18] {\footnotesize $1-p$};
    
    \end{tikzpicture}
\caption{}
\label{fig:toyschemea}
\end{center}
\end{subfigure}
\ifdefined\SINGLECOL
\hfill
\fi
\begin{subfigure}[t]{%
    \ifdefined\SINGLECOL
    0.3\linewidth
    \else
    0.49\linewidth
    \fi
}
\begin{center}
    \begin{tikzpicture}[every node/.style={inner sep=0,outer sep=0}]
    \pgfdeclareimage[width=1.5cm]{bs}{clipart/Objects_tower_B.pdf}
    \pgfdeclareimage[width=0.5cm]{source}{clipart/Objects_Iot_health.pdf}
    \pgfdeclareimage[width=0.4cm]{sensor}{clipart/Objects_Iot_simple.pdf}
    
    \node (s11) at (0,1.5) {\pgfuseimage{source}};
    \node at (0,1.45) {\scriptsize 1};
    \node (s12) at (0,0.5) {\pgfuseimage{source}};
    \node at (0,0.45) {\scriptsize 2};
    
    \node (sensor1) at (3.0,2) {\pgfuseimage{sensor}};
    \node at (3.0,1.85) {\scriptsize $1$};
    \node (sensor2) at (3.0,-0.0) {\pgfuseimage{sensor}};
    \node at (3.0,-0.2) {\scriptsize $2$};
    
    \node[circle,draw,minimum width=1.5em] (state11) [right=0.25 of s11] {\tiny $12$};
    \node[circle,draw,minimum width=1.5em] (state12) [right=0.45 of state11] {\tiny $11$};
    \path[<-] (state11) edge[bend left] (state12);
    \path[<-] (state12) edge[bend left] (state11);
    
    \node[circle,draw,minimum width=1.5em] (state21) [right=0.25 of s12] {\tiny $22$};
    \node[circle,draw,minimum width=1.5em] (state22) [right=0.45 of state21] {\tiny $21$};
    \path[<-] (state21) edge[bend left] (state22);
    \path[<-] (state22) edge[bend left] (state21);
    \node at ($(state11)+(0.55,0.35)$) {\tiny $\alpha$};
    \node at ($(state11)+(0.55,-0.4)$) {\tiny $1-\alpha$};
    \node at ($(state21)+(0.55,0.35)$) {\tiny $1-\alpha$};
    \node at ($(state21)+(0.55,-0.4)$) {\tiny $\alpha$};
    
    \draw (state12) -- (sensor1.west) node[midway,above=0.1,rotate=35] {\footnotesize $p$};
    \draw (state22) -- (sensor2.west) node[midway,below=0.1,rotate=-35] {\footnotesize $1-p$};
    
    \end{tikzpicture}
\caption{}
\label{fig:toyschemeb}
\end{center}
\end{subfigure}
\begin{subfigure}[t]{%
    \ifdefined\SINGLECOL
    0.3\linewidth
    \else
    0.49\linewidth
    \fi
}
\begin{center}
    \begin{tikzpicture}[every node/.style={inner sep=0,outer sep=0}]
    \pgfdeclareimage[width=1.5cm]{bs}{clipart/Objects_tower_B.pdf}
    \pgfdeclareimage[width=0.5cm]{source}{clipart/Objects_Iot_health.pdf}
    \pgfdeclareimage[width=0.4cm]{sensor}{clipart/Objects_Iot_simple.pdf}
    
    \node (s11) at (0,1.5) {\pgfuseimage{source}};
    \node at (0,1.45) {\scriptsize 1};
    \node (s12) at (0,0.5) {\pgfuseimage{source}};
    \node at (0,0.45) {\scriptsize 2};
    
    \node (sensor1) at (3.5,2.0) {\pgfuseimage{sensor}};
    \node at (3.5,1.85) {\scriptsize $1$};
    \node (sensor2) at (3.5,-0.0) {\pgfuseimage{sensor}};
    \node at (3.5,-0.2) {\scriptsize $2$};
    \node (sensor3) at (3.5,1.0) {\pgfuseimage{sensor}};
    \node at (3.5,0.8) {\scriptsize $3$};
    
    \node[circle,draw,minimum width=1.5em] (state11) [right=0.25 of s11] {\tiny $12$};
    \node[circle,draw,minimum width=1.5em] (state12) [right=0.45 of state11] {\tiny $11$};
    \path[<-] (state11) edge[bend left] (state12);
    \path[<-] (state12) edge[bend left] (state11);
    
    \node[circle,draw,minimum width=1.5em] (state21) [right=0.25 of s12] {\tiny $22$};
    \node[circle,draw,minimum width=1.5em] (state22) [right=0.45 of state21] {\tiny $21$};
    \path[<-] (state21) edge[bend left] (state22);
    \path[<-] (state22) edge[bend left] (state21);
    \node at ($(state11)+(0.55,0.35)$) {\tiny $\alpha$};
    \node at ($(state11)+(0.55,-0.4)$) {\tiny $1-\alpha$};
    \node at ($(state21)+(0.55,0.35)$) {\tiny $1-\alpha$};
    \node at ($(state21)+(0.55,-0.4)$) {\tiny $\alpha$};
    
    \draw (state12) -- (sensor1.west) node[midway,above=0.1,rotate=35] {\footnotesize $p$};
    \draw (state22) -- (sensor2.west) node[midway,below=0.1,rotate=-35] {\footnotesize $p$};
    \draw (state12) -- (sensor3.west) node[midway,above=0.1,rotate=-16] {\footnotesize $1-p$};
    \draw (state22) -- (sensor3.west) node[midway,below=0.1,rotate=25] {\footnotesize $1-p$};

    \end{tikzpicture}
\caption{}
\label{fig:toyschemec}
\end{center}
\end{subfigure}
    \caption{Toy scenarios considered in the evaluation.}
    \label{fig:toyscenario}
\end{figure}

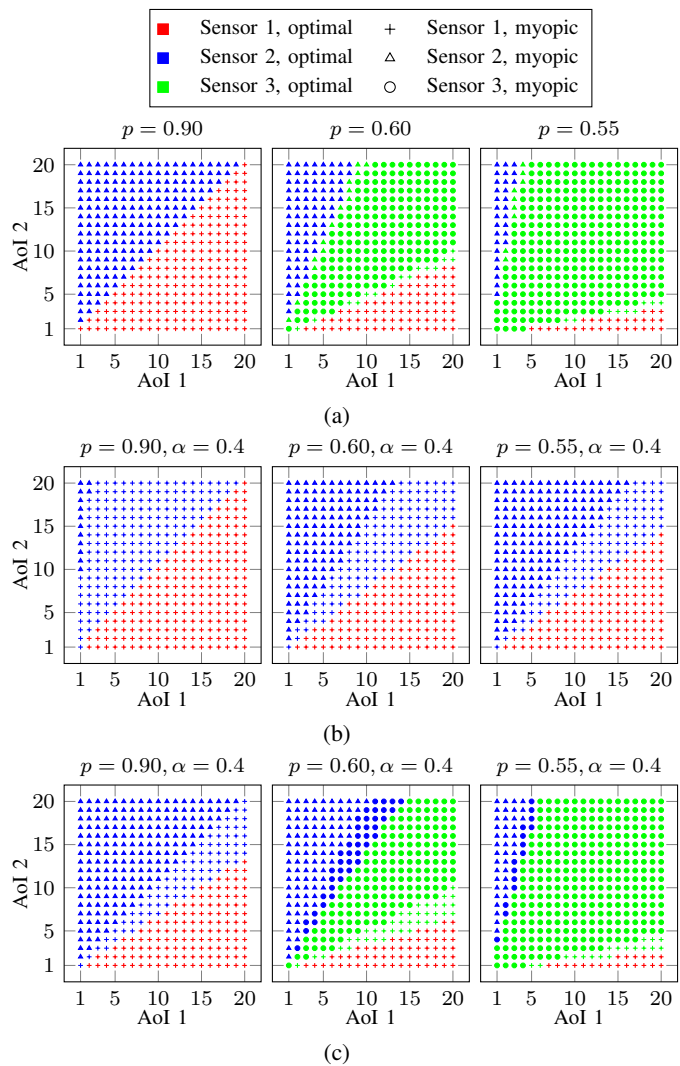
\begin{figure}
  \centering  
    \begin{subfigure}[t]{\linewidth}
    \begin{center}
    \begin{tikzpicture}[tight background]
\begin{groupplot}[group style={
group size= 3 by 1,
horizontal sep=0.15cm,
every plot/.style={%
xtick={1,5,10,15,20},
ytick={1,5,10,15,20},
ylabel shift=-5pt,
xlabel shift=-5pt,
every axis title shift=0pt,
mark size=0.9pt,
scatter/classes={%
0={mark=+,red},%
1={mark=+,blue},%
2={mark=+,green},%
3={mark=triangle*,red},%
4={mark=triangle*,blue},%
5={mark=triangle*,green},%
6={mark=*,red},%
7={mark=*,blue},%
8={mark=*,green}}
}},
height=4.2cm,width=4.2cm,
]

\nextgroupplot[title={$p=0.90$},ylabel={AoI 2},xlabel={AoI 1},
font=\footnotesize,
legend style={column sep=10pt, legend columns=-1,
font=\footnotesize,
legend columns=2,
legend entries={{Sensor 1, optimal}, {Sensor 1, myopic}, {Sensor 2, optimal}, {Sensor 2, myopic}, {Sensor 3, optimal}, {Sensor 3, myopic}},
legend to name=toyscenario_plots_legend}
]

\addlegendimage{only marks, mark=square*,red}
\addlegendimage{only marks, mark=+}
\addlegendimage{only marks, mark=square*,blue}
\addlegendimage{only marks, mark=triangle}
\addlegendimage{only marks, mark=square*,green}
\addlegendimage{only marks, mark=o}

\addplot[scatter,only marks,scatter src=explicit symbolic] table [col sep=comma, meta index=2] {data/toyscenario_plots1_p090.csv};

\nextgroupplot[title={$p=0.60$},xlabel={AoI 1},yticklabels={,,},font=\footnotesize]
\addplot[scatter,only marks,scatter src=explicit symbolic] table [col sep=comma, meta index=2] {data/toyscenario_plots1_p060.csv};

\nextgroupplot[title={$p=0.55$},xlabel={AoI 1},yticklabels={,,},font=\footnotesize]
\addplot[scatter,only marks,scatter src=explicit symbolic] table [col sep=comma, meta index=2] {data/toyscenario_plots1_p055.csv};

\end{groupplot}
\node at ($(group c2r1.north) + (0,1.2cm)$) {\ref{toyscenario_plots_legend}}; 
\end{tikzpicture}\vspace{-1.5em}%
    \caption{}
    \label{fig:toyschemea_plot}
    \end{center}
    \end{subfigure}
    \begin{subfigure}[t]{\linewidth}
    \begin{center}
    \begin{tikzpicture}[tight background]
\begin{groupplot}[group style={
group size= 3 by 1,
horizontal sep=0.15cm,
every plot/.style={%
xtick={1,5,10,15,20},
ytick={1,5,10,15,20},
ylabel shift=-5pt,
xlabel shift=-5pt,
every axis title shift=0pt,
mark size=0.9pt,
scatter/classes={%
0={mark=+,red},%
1={mark=+,blue},%
2={mark=+,green},%
3={mark=triangle*,red},%
4={mark=triangle*,blue},%
5={mark=triangle*,green},%
6={mark=*,red},%
7={mark=*,blue},%
8={mark=*,green}}
}},
height=4.2cm,width=4.2cm,
]

\nextgroupplot[title={$p=0.90, \alpha=0.4$},ylabel={AoI 2},xlabel={AoI 1},
font=\footnotesize]

\addlegendimage{only marks, mark=square*,red}
\addlegendimage{only marks, mark=square*,blue}
\addlegendimage{only marks, mark=+}
\addlegendimage{only marks, mark=triangle}

\addplot[scatter,only marks,scatter src=explicit symbolic] table [col sep=comma, meta index=2] {data/toyscenario_plots2_alpha04_p090.csv};

\nextgroupplot[title={$p=0.60, \alpha=0.4$},xlabel={AoI 1},yticklabels={,,},font=\footnotesize]
\addplot[scatter,only marks,scatter src=explicit symbolic] table [col sep=comma, meta index=2] {data/toyscenario_plots2_alpha04_p060.csv};

\nextgroupplot[title={$p=0.55, \alpha=0.4$},xlabel={AoI 1},yticklabels={,,},font=\footnotesize]
\addplot[scatter,only marks,scatter src=explicit symbolic] table [col sep=comma, meta index=2] {data/toyscenario_plots2_alpha04_p055.csv};

\end{groupplot}
\end{tikzpicture}\vspace{-1.5em}%
    \caption{}
    \label{fig:toyschemeb_plot}
    \end{center}
    \end{subfigure}
    \begin{subfigure}[t]{\linewidth}
    \begin{center}
    \begin{tikzpicture}[tight background]
\begin{groupplot}[group style={
group size= 3 by 1,
horizontal sep=0.15cm,
every plot/.style={%
xtick={1,5,10,15,20},
ytick={1,5,10,15,20},
ylabel shift=-5pt,
xlabel shift=-5pt,
every axis title shift=0pt,
mark size=0.9pt,
scatter/classes={%
0={mark=+,red},%
1={mark=+,blue},%
2={mark=+,green},%
3={mark=triangle*,red},%
4={mark=triangle*,blue},%
5={mark=triangle*,green},%
6={mark=*,red},%
7={mark=*,blue},%
8={mark=*,green}}
}},
height=4.2cm,width=4.2cm,
]

\nextgroupplot[title={$p=0.90, \alpha=0.4$},ylabel={AoI 2},xlabel={AoI 1},
font=\footnotesize]

\addlegendimage{only marks, mark=square*,red}
\addlegendimage{only marks, mark=square*,blue}
\addlegendimage{only marks, mark=square*,green}
\addlegendimage{only marks, mark=+}
\addlegendimage{only marks, mark=triangle}
\addlegendimage{only marks, mark=o}

\addplot[scatter,only marks,scatter src=explicit symbolic] table [col sep=comma, meta index=2] {data/toyscenario_plots3_alpha04_p090.csv};

\nextgroupplot[title={$p=0.60, \alpha=0.4$},xlabel={AoI 1},yticklabels={,,},font=\footnotesize]
\addplot[scatter,only marks,scatter src=explicit symbolic] table [col sep=comma, meta index=2] {data/toyscenario_plots3_alpha04_p060.csv};

\nextgroupplot[title={$p=0.55, \alpha=0.4$},xlabel={AoI 1},yticklabels={,,},font=\footnotesize]
\addplot[scatter,only marks,scatter src=explicit symbolic] table [col sep=comma, meta index=2] {data/toyscenario_plots3_alpha04_p055.csv};

\end{groupplot}
\end{tikzpicture}\vspace{-1.5em}%
    \caption{}
    \label{fig:toyschemec_plot}
    \end{center}
    \end{subfigure}
    \caption{Scheduling policies for the toy scenarios in \cref{fig:toyscenario}. The optimal actions obtained using relative value iteration are indicated by the color, while the shapes of the points indicate the myopic actions.}
    \label{fig:toyscenario_plots}
\end{figure}

For simplicity, we assume a perfect channel between the sensor and the destination, i.e., $q_1=q_2=1$. We limit ourselves to consider the case when source 1 is in state $11$ and source 2 is in state $21$, which represents the only non-trivial scheduling situation in all of the three scenarios. The value functions used by the optimal policy are obtained with a truncation of the AoI at $Q=100$.
The corresponding optimal policies are shown in \cref{fig:toyscenario_plots}, where the colors of the points indicate the optimal policy obtained by relative value iteration and the point shapes show the myopic action. The policy for the scenario in \cref{fig:toyschemea} is shown in \cref{fig:toyschemea_plot}. It can be seen that the optimal policy is to schedule sensors 1 and 2 when $p$ is high or when the AoI of one source is much higher than the AoI of the other. On the other hand, sensor 3 is advantageous when $p$ is low or the AoI difference is small. Note that the actions by the myopic policy are very close to the optimal ones only with differences close to the decision boundaries. While this subtle difference is unlikely to have a significant impact on the performance, it is due to the memory introduced by the AoI process, causing the scheduling problem is non-trivial. As established by \cref{prop:myopicoptimal}, the difference between the myopic and optimal policies vanishes as $p$ approaches 1, which is also indicated by the figures already for $p=0.9$ where the myopic and optimal policies coincide in considered range of AoIs.

The difference between the optimal and the myopic policies become more pronounced when the sources have states that impact their observability. In the scenario in \cref{fig:toyschemeb}, source 1 can be observed by a sensor only in a fraction $1-\alpha$ of the time, while source 2 is in a fraction $\alpha$ of the time. As a result, for $\alpha=0.4$ shown in \cref{fig:toyschemeb_plot}, the optimal policy is much more likely to schedule sensor 2 than the myopic policy is. The same can be seen in the last scenario in \cref{fig:toyschemec}, where the optimal policy also favors sensor 2, while the myopic policy is the same as in the first scenario.

\subsection{Small Factory Scenario}
We now return to the scenario from the introduction where AGVs (sources) move around in four zones (represented by their states), out of which three are covered by cameras (sensors). We study the average AoI achieved by the derived policies and investigate the AoI penalty that is achieved by partial observability (with detectable sources) and no delay ($\tau=0$). We assume again that there are three AGVs, and denote by $\alpha$ the probability that they move from one zone to an adjacent one, so that they remain in the same zone with probability $1-2\alpha$ (they cannot move diagonally). The probability that a camera fails to observe an AGV that is within its zone, e.g., if the AGV is hidden behind an object, is denoted by $1-p$, and we assume again no transmission errors, i.e., $q_n=1$ for all $n$. The value functions used by all except the uniform random and myopic policies are obtained from relative value iteration truncated at an AoI of $Q=20$.

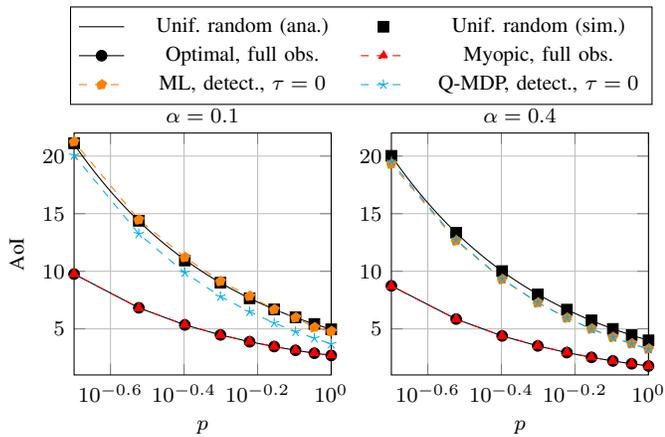
\begin{figure}
    \centering
    \begin{tikzpicture}[tight background]
  \begin{groupplot}[group style={%
      group size= 2 by 1,
      ylabels at=edge left,
  horizontal sep=0.8cm,
},
title style={yshift=-0.2cm},
height=4.8cm,width=5cm,
xmin=0.2, xmax=1,
ymin=1, ymax=22,
every axis/.append style={mark options=solid},
]

\nextgroupplot[xmode=log,title={$\alpha=0.1$},ylabel={AoI},xlabel={$p$},
font=\footnotesize,grid=both,
legend style={column sep=10pt,
font=\footnotesize,
legend columns=2,
legend to name=small_cam_scenario_legend}
]
\addplot+[mark=none,color=black] table [col sep=comma,x index=0,y index=1] {data/small_cam_0.1_rnd_ana.csv};
\addlegendentry{Unif. random (ana.)};

\addplot+[only marks,mark=square*,color=black,mark options={fill=black}] table [col sep=comma,x index=0,y index=1] {data/small_cam_0.1.csv};
\addlegendentry{Unif. random (sim.)};

\addplot+[mark=*,color=black,mark options={fill=black}] table [col sep=comma,x index=0,y index=2] {data/small_cam_0.1.csv};
\addlegendentry{Optimal, full obs.};

\addplot+[dashed, mark=triangle*,color=red] table [col sep=comma,x index=0,y index=3] {data/small_cam_0.1.csv};
\addlegendentry{Myopic, full obs.};

\addplot+[dashed, mark=pentagon*, color=orange,mark options={fill=orange}] table [col sep=comma,x index=0,y index=4] {data/small_cam_0.1.csv};
\addlegendentry{ML, detect., $\tau=0$};
\addplot+[dashed,mark=star,color=cyan,mark options={fill=cyan}] table [col sep=comma,x index=0,y index=5] {data/small_cam_0.1.csv};
\addlegendentry{Q-MDP, detect., $\tau=0$};

\nextgroupplot[xmode=log,title={$\alpha=0.4$},xlabel={$p$},font=\footnotesize,grid=both]

\addplot+[mark=none,color=black] table [col sep=comma,x index=0,y index=1] {data/small_cam_0.4_rnd_ana.csv};

\addplot+[only marks,mark=square*,color=black,mark options={fill=black}] table [col sep=comma,x index=0,y index=1] {data/small_cam_0.4.csv};

\addplot+[mark=*,color=black,mark options={fill=black}] table [col sep=comma,x index=0,y index=2] {data/small_cam_0.4.csv};

\addplot+[dashed, mark=triangle*,color=red] table [col sep=comma,x index=0,y index=3] {data/small_cam_0.4.csv};

\addplot+[dashed, mark=pentagon*, color=orange,mark options={fill=orange}] table [col sep=comma,x index=0,y index=4] {data/small_cam_0.4.csv};
\addplot+[dashed,mark=star,color=cyan,mark options={fill=cyan}] table [col sep=comma,x index=0,y index=5] {data/small_cam_0.4.csv};

\end{groupplot}
\path (group c1r1.north west) -- node[above=0.25cm]{\ref{small_cam_scenario_legend}} (group c2r1.north east);
\end{tikzpicture}%
    \caption{Average AoI achieved by the different policies in the small factory scenario.}
    \label{fig:small_cam_scenario_plots}
\end{figure}

The average AoI of five different policies are shown in \cref{fig:small_cam_scenario_plots} for $\alpha=0.1$ and $\alpha=0.4$, respectively. In both cases, the myopic policy performs close to optimally despite the fact that the sources are stateful, indicating that the additional cost of acting myopically is negligible in terms of average AoI. The AoI in the case of $\alpha=0.1$ is generally higher than for $\alpha=0.4$, which is because the sources stay longer in the invisible zone.
Regarding the policies with partial observability, it can be seen in that for $\alpha=0.1$ the ML policy is only little if at all better than the random policy, while the Q-MDP policy is significantly better. This is because for low values of $\alpha$ where the sources are most likely to stay in the same zone between time slots, the ML policy tends to conclude that a given source is in the hidden zone. Since no sensor can observe this zone, it can do nothing but act like the random policy. On the other hand, the Q-MDP policy is able to make use of the entire belief distribution, and not only the most likely belief, which allows it to take more informed scheduling decisions using the less likely states.
However, as $\alpha$ increases, the AoI of the ML policy gets closer to that of the Q-MDP policy (see the right plot), which is because the likelihood of staying in the same zone for a long time decreases, and thus the maximum likelihood belief is more likely to be outside the hidden zone.

\subsection{Large Factory Scenario}
Finally, we consider a large factory scenario comprising an $8\times 8$ grid of 64 cells, and 10 AGVs. Each of the cells is equipped with a camera that observes the AGVs inside the zone with probability one. However, in addition to the 64 cameras that cover the individual cells, there are sensors that cover larger areas but with less reliable observations. Specifically, in addition to the 64 sensors that cover each cell, there are 16 sensors that cover each of the $2\times 2$ zones, 4 sensors that cover each of the $4\times 4$ zones, and one sensor that covers the entire area. We refer to these groups of sensors as level 1--4 sensors, where level 1 sensors are the sensors that cover a single cell, and the level 4 sensor is the one that covers the entire area. The observation probabilities for the sensors at a certain level, $l=1,\ldots,4$, are defined as
\begin{equation}
    p_l=\gamma^{l-1},
\end{equation}
where $\gamma$ is a degradation factor that controls how much the observation probability decreases per level. For instance, if $\gamma=1$, the sensors at each level observe each of the cells within their zones with probability $1$, whereas if $\gamma=1/4$ the observation probability at each level is the reciprocal of the number of cells that the sensors cover. Thus, when $\gamma$ is high the scheduler is more likely to schedule the higher level sensors than when $\gamma$ is low. By controlling the probability that a sensor detects multiple sources, we can obtain general insight into the impact of the ability to detect multiple sources on the policies and the AoI. As before $\alpha$ denotes the probability of transitioning into an adjacent cell. Due to the size of the problem, we consider only myopic policies.

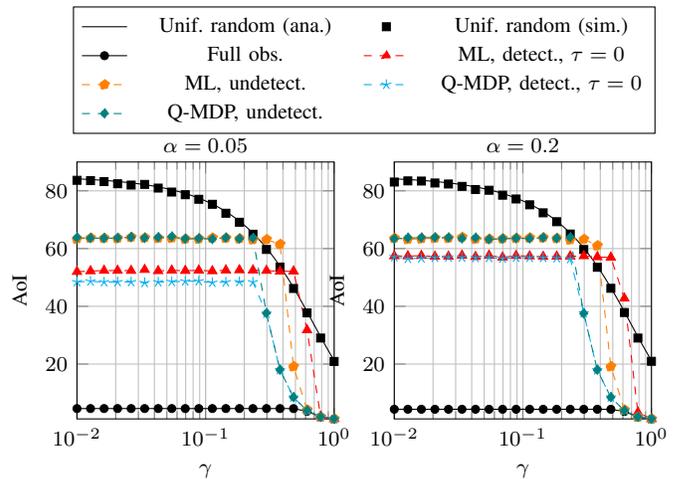
\begin{figure}
    \centering
    \begin{tikzpicture}[tight background]
\begin{groupplot}[group style={%
  \ifdefined\SINGLECOL
  group size= 2 by 1,
  horizontal sep=2cm,
  \else
  group size= 2 by 1,
  ylabels at=edge left,
  horizontal sep=0.8cm,
  \fi
  vertical sep=1.5cm
},
title style={yshift=-0.2cm},
height=5cm,width=5cm,
xmin=0.01, xmax=1,
ymin=1, ymax=90,
every axis/.append style={mark options=solid},
]

\nextgroupplot[xmode=log,title={$\alpha=0.05$},ylabel={AoI},xlabel={$\gamma$},
font=\footnotesize,grid=both,
legend style={column sep=10pt,
font=\footnotesize,
legend columns=2,
legend to name=large_cam_scenario_legend}
]
\addplot+[mark=none,color=black] table [col sep=comma,x index=0,y index=1] {data/large_cam_0.05_rnd_ana.csv};
\addlegendentry{Unif. random (ana.)};

\addplot+[only marks,mark=square*,color=black,mark options={fill=black},mark size=1.5pt] table [col sep=comma,x index=0,y index=1] {data/large_cam_0.05.csv};
\addlegendentry{Unif. random (sim.)};

\addplot+[mark=*,color=black,mark options={fill=black},mark size=1.5pt] table [col sep=comma,x index=0,y index=2] {data/large_cam_0.05.csv};
\addlegendentry{Full obs.};

\addplot+[dashed, mark=triangle*,color=red] table [col sep=comma,x index=0,y index=3] {data/large_cam_0.05.csv};
\addlegendentry{ML, detect., $\tau=0$};

\addplot+[dashed, mark=pentagon*, color=orange,mark options={fill=orange}] table [col sep=comma,x index=0,y index=4] {data/large_cam_0.05.csv};
\addlegendentry{ML, undetect.};
\addplot+[dashed,mark=star,color=cyan,mark options={fill=cyan}] table [col sep=comma,x index=0,y index=5] {data/large_cam_0.05.csv};
\addlegendentry{Q-MDP, detect., $\tau=0$};
\addplot+[dashed,mark=diamond*,color=teal,mark options={fill=teal}] table [col sep=comma,x index=0,y index=6] {data/large_cam_0.05.csv};
\addlegendentry{Q-MDP, undetect.};

\nextgroupplot[xmode=log,title={$\alpha=0.2$},ylabel={AoI},xlabel={$\gamma$},font=\footnotesize,grid=both]
\addplot+[mark=none,color=black] table [col sep=comma,x index=0,y index=1] {data/large_cam_0.2_rnd_ana.csv};

\addplot+[only marks,mark=square*,color=black,mark options={fill=black},mark size=1.5pt] table [col sep=comma,x index=0,y index=1] {data/large_cam_0.2.csv};

\addplot+[mark=*,color=black,mark options={fill=black},mark size=1.5pt] table [col sep=comma,x index=0,y index=2] {data/large_cam_0.2.csv};

\addplot+[dashed, mark=triangle*,color=red] table [col sep=comma,x index=0,y index=3] {data/large_cam_0.2.csv};

\addplot+[dashed, mark=pentagon*, color=orange,mark options={fill=orange}] table [col sep=comma,x index=0,y index=4] {data/large_cam_0.2.csv};
\addplot+[dashed,mark=star,color=cyan,mark options={fill=cyan}] table [col sep=comma,x index=0,y index=5] {data/large_cam_0.2.csv};
\addplot+[dashed,mark=diamond*,color=teal,mark options={fill=teal}] table [col sep=comma,x index=0,y index=6] {data/large_cam_0.2.csv};

\end{groupplot}
\ifdefined\SINGLECOL
\path (group c1r1.north west) -- node[above=0.25cm]{\ref{large_cam_scenario_legend}} (group c2r1.north east);
\else
\path (group c1r1.north west) -- node[above=0.25cm]{\ref{large_cam_scenario_legend}} (group c2r1.north east);
\fi
\end{tikzpicture}\vspace{-1.0em}%
    \caption{Average AoI for the large industrial scenario obtained using myopic policies for $\alpha=0.05$ and $\alpha=0.2$.}
    \label{fig:large_cam_scenario_plots}
\end{figure}

\Cref{fig:large_cam_scenario_plots} shows the average AoI obtained using various policies for $\alpha=0.05$ and $\alpha=0.2$ with no detection delay, i.e., $\tau=0$.
In general, the average AoI when $\alpha=0.05$ is lower than for $\alpha=0.2$, because the AGVs move slower, and thus are more predictable.
In both cases, the average AoI is constant for small values of $\gamma$ (except for the uniform random policy), which reflects that the policy is the same. Specifically, all policies sample only sensors at level one. However, when $\gamma$ approaches $\gamma=1/4$, it becomes beneficial to schedule higher level sensors, either when multiple AGVs are likely to be located in the same area, or when the uncertainty is high and scheduling a level one sensor is unlikely to result in an observation. Note that this effect is first reflected for the Q-MDP policies. This is because they, contrary to the ML policies, include the uncertainty in their decisions. In fact, when $\gamma\ge 1/4$ Q-MDP agents only schedule the sensor at level 4, while the ML agents require a larger value of $\gamma$ before they start scheduling higher level sensors. Furthermore, this transition is independent of $\alpha$, suggesting that, regardless of the system dynamics, the ability for sensors to detect multiple sources is particularly useful when there is uncertainty about the source states.

Somewhat surprisingly, the ML policy with undetectable sources leads to a faster AoI decrease compared to the ML policy with detectable sources, and performs better within a range of $\gamma$. The reason is that in the detectable case where the AoIs are known to the agent, the policy is likely to be dictated by a few AGVs with high AoI, which are scheduled at the grid level. On the other hand, in the undetectable case where the AoIs are unobservable, the AoI beliefs of the different AGVs are much more likely to be concentrated around the same values, which makes the agent more likely to schedule the sensors at the higher levels.

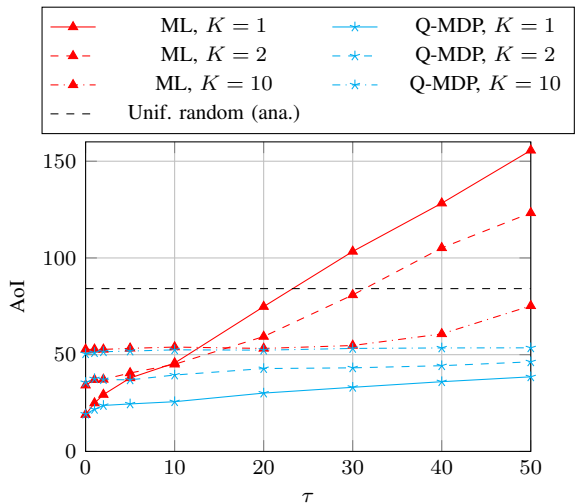
\begin{figure}
  \centering
  \begin{tikzpicture}[tight background]
\begin{axis}[
height=5.7cm,width=7.5cm,
every axis/.append style={mark options=solid},
ylabel={AoI},
xlabel={$\tau$},
font=\footnotesize,
grid=both,
xmin=0,xmax=50,
ymin=0,ymax=160,
legend style={column sep=10pt, font=\footnotesize,legend columns=2, at={(0.5,1.025)},anchor=south}
]

\addplot+[mark=triangle*,color=red,mark options={fill=red}] table [col sep=comma,x index=0,y index=2] {data/large_cam_delayed.csv};
\addlegendentry{ML, $K=1$};

\addplot+[mark=star,color=cyan,mark options={fill=cyan}] table [col sep=comma,x index=0,y index=3] {data/large_cam_delayed.csv};
\addlegendentry{Q-MDP, $K=1$};

\addplot+[dashed, mark=triangle*,color=red,mark options={fill=red,solid}] table [col sep=comma,x index=0,y index=5] {data/large_cam_delayed.csv};
\addlegendentry{ML, $K=2$};

\addplot+[dashed,mark=star,color=cyan,mark options={fill=cyan,solid}] table [col sep=comma,x index=0,y index=6] {data/large_cam_delayed.csv};
\addlegendentry{Q-MDP, $K=2$};

\addplot+[dashdotted, mark=triangle*,color=red,mark options={fill=red,solid}] table [col sep=comma,x index=0,y index=8] {data/large_cam_delayed.csv};
\addlegendentry{ML, $K=10$};

\addplot+[dashdotted,mark=star,color=cyan,mark options={fill=cyan,solid}] table [col sep=comma,x index=0,y index=9] {data/large_cam_delayed.csv};
\addlegendentry{Q-MDP, $K=10$};

\addplot[color=black,dashed,no marks] table [col sep=comma,x index=0,y index=1] {data/large_cam_delayed.csv};
\addlegendentry{Unif. random (ana.)};

\end{axis}
\end{tikzpicture}%
  \caption{Average AoI for the large industrial scenario with delayed detectable sources for $\alpha=0.05$, $\gamma=0.01$, and $K$ AGVs.}
  \label{fig:large_cam_scenario_delayed_plot}
\end{figure}

We finish the section by considering the effect of the detection delay $\tau$. Because the impact of a delay is most prominent for small $\alpha$ and a small number of AGVs, $K$, we limit the study to the case with $\alpha=0.05$ and $K\in\{1,2,10\}$, and fix $\gamma=0.01$. The average AoI is shown in \cref{fig:large_cam_scenario_delayed_plot} for various delays $\tau$. It can be seen that the delay comes at a relatively small penalty for the Q-MDP policy, while the ML policy incurs a large penalty when the delay is large. This is because the ML policy, contrary to the Q-MDP policy, only considers the most likely state. When the probability that a scheduled camera observes an AGV is small, the maximum likelihood state is almost unaffected by the unobserved scheduling decisions. The Q-MDP policy, on the other hand, exploits the entire scheduling history to avoid scheduling the same camera many times in a row.

\section{Conclusion}\label{sec:conclusion}
Timeliness of information is a relevant performance indicator for many IoT applications in which the AoI serves as an attractive alternative to the traditional end-to-end latency by considering the information from the destination perspective as opposed to the transmitter perspective. In this paper, we have considered the problem of scheduling a set of sensors that observe a set of dynamic sources. Furthermore, we have assumed that each source is observed by multiple sensors with probabilities that depend on the current state of the source. This scenario represents a case that arises naturally in many IoT monitoring applications, such as camera monitoring. We have formulated the scheduling problem as an MDP, and used it to derive optimal and approximate, but more practical, scheduling policies for both the fully observable and the partially observable cases. Through numerical results, we have shown that due to the dynamic sources, partial observability comes at a high cost in terms of AoI compared to the fully observable case. Furthermore, we have shown that the fact that sensors can observe multiple sources can be beneficial especially in the partial observable case where there is uncertainty about the source states. However, this requires that the policy is able to take this into account, which is the case for the approximate Q-MDP policy, but not for the ML policy. On the other hand, the ML policy has a lower complexity, which may be advantageous in some situations.

\appendices
\crefalias{section}{appendix}
\section{Proof of \cref{theo:randompolicyavgaoi}}\label{app:randompolicyavgaoiproof}
We first note that the sources are independent, and thus the average AoI of each source can be computed independently, and the total average AoI can be obtained as the arithmetic mean of the average AoIs of the individual sources.

The AoI of a single source evolves according to a two-dimensional Markov chain in which one dimension represents its AoI and the other represents its internal state. The probability that the source is observed in a given time slot depends on its current state and the sensor that was scheduled. Since exactly one of the $N$ sensors is scheduled, drawn from the discrete distribution $\Pr(a_t=n)=\bar{a}_n$, the conditional probability that source $k$ is observed given its state can be represented by the vector
\begin{equation}
    \mathbf{p}_k=\sum_{n=1}^N \bar{a}_n q_n\begin{bmatrix}p_{nk}^{(1)} & p_{nk}^{(2)} &\cdots & p_{nk}^{(S_k)}\end{bmatrix},
\end{equation}
where each entry represents one of the $S_k$ source states. Using this, the \emph{conditional} transition matrix given that the source is observed can be written $\mathbf{R}_k^{(\text{succ})}=\diag(\mathbf{p}_k)\mathbf{R}_k$. Similarly, the conditional transition matrix given that the source is \emph{not} observed is $\mathbf{R}_k^{(\text{fail})}=\left(\mathbf{I}-\diag(\mathbf{p}_k)\right)\mathbf{R}_k$. Using this, the truncated joint transition matrix of the source AoI and its state with a maximal AoI of $Q$, $\widehat{\bm{\Psi}}_k\in\mathbb{R}^{S_k Q\times S_k Q}$, is given as
\begin{equation}
    \widehat{\bm{\Psi}}_k=\begin{bmatrix}
    \mathbf{R}_k^{(\text{succ})} & \mathbf{R}_k^{(\text{fail})} & 0  &\cdots & 0\\
    \mathbf{R}_k^{(\text{succ})} & 0 & \mathbf{R}_k^{(\text{fail})} & \cdots & 0 \\
    \vdots & \vdots & \vdots & \ddots & \vdots\\
    \mathbf{R}_k^{(\text{succ})} & 0 & 0 & \cdots & \mathbf{R}_k^{(\text{fail})} \\
    \mathbf{R}_k^{(\text{succ})} & 0 & 0 & \cdots & \mathbf{R}_k^{(\text{fail})}
    \end{bmatrix}.
\end{equation}
It is easy to verify that because $\mathbf{R}_k$ defines an irreducible and aperiodic Markov chain, the Markov chain defined by $\widehat{\bm{\Psi}}_k$ is irreducible, aperiodic and positive recurrent and thus has a unique stationary distribution given by the vector $\bm{\phi}_k$ that satisfies
\begin{equation}\label{eq:phi_stationaryequation}
    \bm{\phi}_k \widehat{\bm{\Psi}}_k = \bm{\phi}_k.
\end{equation}
Due to the structure of $\widehat{\bm{\Psi}}_k$, $\bm{\phi}_k$ is a (block) vector
\begin{equation}
    \bm{\phi}_k=\begin{bmatrix}
    \bm{\phi}_k^{(1)}&\bm{\phi}_k^{(1)}&\cdots&\bm{\phi}_k^{(Q)}
    \end{bmatrix},
\end{equation}
where $\bm{\phi}_k^{(q)}=\begin{bmatrix}\Pr(q,s_k=1)&\cdots&\Pr(q,s_k=K)\end{bmatrix}$ is the $S_k$-element vector where the $j$-th entry is the joint probability that source $k$ is in state $j$ and its AoI is $q$.

By writing out \cref{eq:phi_stationaryequation} it is straightforward to show that the blocks of $\bm{\phi}_k$ can be written
\begin{equation}
    \bm{\phi}_k^{(q)}=\bm{\beta}_k\mathbf{R}_{k}^{(\text{succ})}\left(\mathbf{R}_{k}^{(\text{fail})}\right)^{q-1},
\end{equation}
and the average Q-truncated AoI of source $k$ is thus
\begin{align}
    \bar{\Delta}_k^{(Q)}&=\sum_{q=1}^Q q \bm{\phi}_k^{(q)}\mathbf{1}^T\\
    &=\bm{\beta}_k\mathbf{R}_{k}^{(\text{succ})}\sum_{q=1}^Q q\left(\mathbf{R}_{k}^{(\text{fail})}\right)^{q-1}\mathbf{1}^T
\end{align}
where $\mathbf{1}$ is the row vector of all ones of appropriate dimension.
Provided that the source can be observed by a sensor that is scheduled with a non-zero probability, i.e., $\sum_{n=1}^N\sum_{s=1}^{S_k}\bar{a}_n q_n p_{nk}^{(s)}>0$, it follows from the Perron-Frobenius theorem~\cite{pillai05} that the eigenvalues of $\mathbf{R}_{k}^{(\text{fail})}$ are non-negative real and strictly less than one. Thus, the limit $\lim_{Q\to\infty}\sum_{q=1}^Q q \left(\mathbf{R}_{k}^{(\text{fail})}\right)^{q-1}$ exists and the untruncated AoI is
\begin{align}
    E[\Delta_{\text{random},k}]&=\lim_{Q\to\infty}\bar{\Delta}_k^{(Q)}\\
    &=\bm{\beta}_k\mathbf{R}_{k}^{(\text{succ})}\sum_{q=1}^{\infty} q\left(\mathbf{R}_{k}^{(\text{fail})}\right)^{q-1}\mathbf{1}^T\\
    &=\bm{\beta}_k\mathbf{R}_{k}^{(\text{succ})}\left(\sum_{q=1}^{\infty} \left(\mathbf{R}_{k}^{(\text{fail})}\right)^{q-1}\right)^2\mathbf{1}^T\\
    &=\bm{\beta}_k\mathbf{R}_{k}^{(\text{succ})}\left(\mathbf{I}-\mathbf{R}_{k}^{(\text{fail})}\right)^{-2}\mathbf{1}^T.
\end{align}
Averaging over the $K$ sources completes the proof.

\section{Proof of \cref{theo:optstationarypolicy}}\label{app:proofoptstationarypolicy}
We prove \cref{theo:optstationarypolicy} by showing that the sufficient conditions in~\cite{sennott89} are fulfilled, by following a similar approach as in~\cite{hsu2017}. We first state the conditions in a number of lemmas and then prove \cref{theo:optstationarypolicy}.
\begin{lemma}\label{pro:condition1}
There exists a deterministic stationary policy that minimizes the discounted cost $V_{\alpha}(\bm{\Lambda})$ for any $0<\alpha<1$. Furthermore, $V_{\alpha}(\bm{\Lambda})$ is finite for every $\alpha$ and $\bm{\Lambda}$, and satisfies the Bellman equation
\ifdefined\SINGLECOL
\begin{align} 
    V_{\alpha}(\bm{\Lambda}) &= \min_{a}\Bigg\{ C(\bm{\Lambda}, a) + \alpha\sum_{\bm{\Lambda}_{t+1}}\Pr(\bm{\Lambda}_{t+1}|\bm{\Lambda}_t=\bm{\Lambda}, a_t=a)V_{\alpha}(\bm{\Lambda}_{t+1})\Bigg\}\nonumber\\
    &=\min_{a}\left\{ C(\bm{\Lambda}, a) + \alpha\E_{\bm{\Lambda}'}\left[V_{\alpha}(\bm{\Lambda}')|\bm{\Lambda},a\right]\right\}.
\end{align}
\else
\begin{align} 
\begin{split}
    V_{\alpha}(\bm{\Lambda}) &= \min_{a}\Bigg\{ C(\bm{\Lambda}, a)\\
    &\qquad+ \alpha\sum_{\bm{\Lambda}_{t+1}}\Pr(\bm{\Lambda}_{t+1}|\bm{\Lambda}_t=\bm{\Lambda}, a_t=a)V_{\alpha}(\bm{\Lambda}_{t+1})\Bigg\}
\end{split}\nonumber\\
    &=\min_{a}\left\{ C(\bm{\Lambda}, a) + \alpha\E_{\bm{\Lambda}'}\left[V_{\alpha}(\bm{\Lambda}')|\bm{\Lambda},a\right]\right\}.
\end{align}
\fi
\end{lemma}
\begin{proof}
Let $C_{\bm{\Lambda},\bm{\Lambda}_0}$ denote the (random) cost of a first passage from $\bm{\Lambda}$ to $\bm{\Lambda}_0$ under some policy. It suffices to show that there exists a stationary policy that induces an irreducible, aperiodic Markov chain satisfying $\E[C_{\bm{\Lambda},\bm{\Lambda}_0}]<\infty$ for all $\bm{\Lambda}$~\cite[Proposition~5]{sennott89}. To show that this condition is satisfied, consider the stationary random policy introduced in \cref{sec:random_policy} that schedules a random sensor in each time slot. Because the induced Markov chains from \cref{app:randompolicyavgaoiproof} describing the cost of each source are irreducible, aperiodic and positive recurrent, the Markov chain describing the average cost of all $K$ sources (obtained by augmenting the state space) is also irreducible, aperiodic and positive recurrent. Let $\bm{\Lambda}_0=(\bm{S}_0, \bm{\Delta}_0)$, where $\bm{S}_0$ is an arbitrary set of source states and $\bm{\Delta}_0=\{1,2,\ldots,K\}$, then, due to the positive recurrence, the expected number of steps required to reach $\bm{\Lambda}_0$ from any state $\bm{\Lambda}$ is finite. Since $C(\bm{\Lambda}, a_t)<\infty$ for any $\bm{\Lambda}$ and $a_t$, we have that $\E[C_{\bm{\Lambda},\bm{\Lambda}_0}]<\infty$.
\end{proof}

\begin{lemma}\label{pro:nonnegative}
$V_{\alpha}(\bm{\Lambda})$ is nondecreasing in the AoI of each source, $\Delta_k$.
\end{lemma}
\begin{proof}
We prove the result by induction using the value iteration procedure in \cref{eq:valueiter}, which is guaranteed to converge to the optimal value function for $0<\alpha<1$ from any initial $V_{\alpha}^{0}(\bm{\Lambda})$.
Note first that the cost function $C(\bm{\Lambda}, a)$ defined in \eqref{eq:costfun} is nondecreasing in $\Delta_k$, and by setting $V_{\alpha}^{0}(\bm{\Lambda})=0$ this also holds for the base case. For the induction step, note that the sum of nondecreasing functions is also nondecreasing, and the minimum operator preserves this property. Therefore, if  $V_{\alpha}^{n}(\bm{\Lambda})$ is nondecreasing in $\Delta_k$, then $V_{\alpha}^{n+1}(\bm{\Lambda})$ is also nondecreasing in $\Delta_k$. It follows that $V_{\alpha}(\bm{\Lambda})=\lim_{n\to\infty}V_{\alpha}^{n}(\bm{\Lambda})$ is nondecreasing in $\Delta_k$ as well.
\end{proof}

\begin{lemma}\label{pro:condition2}
There exists a finite, nonnegative $N$ such that $-N\le h_{\alpha}(\bm{\Lambda})$ for all $\bm{\Lambda}$ and $\alpha$.
\end{lemma}
\begin{proof}
We prove that the inequality holds by first showing that the number of states $\bm{\Lambda}$ with $V_{\alpha}(\bm{\Lambda})< V_{\alpha}(\bm{\Lambda}_0)$ is finite, and then showing that the value function of these states is bounded from below by a constant that is independent of $\alpha$.
Define $\bm{\Lambda}_0=(\bm{S}_0, \bm{\Delta}_0)$, where $\bm{S}_0$ is an arbitrary set of source states and $\bm{\Delta}_0=\{1,1,\ldots,1\}$. Because the value function $V_{\alpha}(\bm{\Lambda})$ is finite and nondecreasing in $\Delta_k$ (as established by \cref{pro:condition1,pro:nonnegative}), any finite $N\ge 0$ satisfies the condition for all states $\bm{\Lambda}=(\bm{S},\bm{\Lambda})$ with $\bm{\Lambda}\neq\bm{\Lambda}_0$. Furthermore, because the set of source states is finite, there are at most a finite number of states $(\bm{S},\bm{\Delta}_0)$
for which $V_{\alpha}((\bm{S},\bm{\Delta}_0))< V_{\alpha}(\bm{\Lambda}_0)$. Because $\bm{S}$ evolves according to an irreducible positive recurrent Markov chain (independent of the policy), any state $(\bm{S},\bm{\Delta}_0)$ will transition into $(\bm{S}_0,\bm{\Delta})$ for some AoI state $\bm{\Delta}$ within finite expected time at a finite expected total cost of $\E[C_{(\bm{S},\bm{\Delta}_0),(\bm{S}_0,\bm{\Delta})}]$. Because the value function is nondecreasing in $\Delta_k$, we have that $V_{\alpha}((\bm{S}_0,\bm{\Delta}))\ge V_{\alpha}(\bm{\Lambda}_0)$. From the definition of the value function, this implies that $V_{\alpha}(\bm{\Lambda}_0)\le V_{\alpha}((\bm{S},\bm{\Delta}_0))+\E[C_{(\bm{S},\bm{\Delta}_0),(\bm{S}_0,\bm{\Delta})}]$ and equivalently $-\E[C_{(\bm{S},\bm{\Delta}_0),(\bm{S}_0,\bm{\Delta})}]\le V_{\alpha}((\bm{S},\bm{\Delta}_0))-V_{\alpha}(\bm{\Lambda}_0)$. Consequently, by setting $N=\max_{\bm{S}}\E[C_{(\bm{S},\bm{\Delta}_0),(\bm{S}_0,\bm{\Delta}')}]$ where $\bm{\Delta}'=\min \{\bm{\Delta} | \max_{\bm{S}}\E[C_{(\bm{S},\bm{\Delta}_0),(\bm{S}_0,\bm{\Delta})}]<\infty\}$ the condition is satisfied.
\end{proof}

\begin{lemma}\label{pro:condition3}
There exist a nonnegative $M_{\bm{\Lambda}}$ such that $h_{\alpha}(\bm{\Lambda})\le M_{\bm{\Lambda}}$ for all $\bm{\Lambda}$ and $\alpha$. Furthermore, $\E\left[M_{\bm{\Lambda}_{t+1}}|\bm{\Lambda}_t,a_t\right]<\infty$ for all $\bm{\Lambda}_t$ and $a_t$.
\end{lemma}
\begin{proof}
It is sufficient to show that for each state $\bm{\Lambda}$ and action $a$, there exists a stationary policy $\pi_{\bm{\Lambda}',a}$ that chooses action $a$ in state $\bm{\Lambda}'$ and induces an irreducible, aperiodic Markov chain that satisfies $\E[C_{\bm{\Lambda},\bm{\Lambda}_0}]<\infty$ for all $\bm{\Lambda}$~\cite[Proposition 5]{sennott89}. To show that this condition is satisfied, let $\pi_{\bm{\Lambda}',a}$ be the random scheduling policy presented in \cref{sec:random_policy} that schedules in each time slot a sensor chosen uniformly at random, except in state $\bm{\Lambda}'$, where it deterministically schedules a specific sensor indicated by the action $a$. Taking action $a$ in state $\bm{\Lambda}'$ comes at a finite cost, and, since the action does not influence the source states but only the AoI, leads to a new state from which all other states can still be reached within finite time. Therefore, the Markov chain induced by $\pi_{\bm{\Lambda}',a}$ is still irreducible, aperiodic and positive recurrent, and thus has a unique stationary distribution. Furthermore, it satisfies $C(\bm{\Lambda}, a_t)<\infty$ for any $\bm{\Lambda}$ and $a_t$, and by repeating the argument used in the proof of \cref{pro:condition1} we conclude that $\E[C_{\bm{\Lambda},\bm{\Lambda}_0}]<\infty$.
\end{proof}

In accordance with~\cite[Theorem]{sennott89}, the conditions presented in \cref{pro:condition1,pro:condition2,pro:condition3} are sufficient for \cref{theo:optstationarypolicy}.

\bibliographystyle{IEEEtran}

\bibliography{IEEEabrv,bibliography}

\end{document}